\documentclass[dvips,generic,preprint]{imsart}
\setattribute{journal}{name}{}
\RequirePackage[OT1]{fontenc}
\RequirePackage{amsthm,amsmath,natbib}
\RequirePackage[colorlinks,citecolor=blue,urlcolor=blue]{hyperref}
\RequirePackage{hypernat}
\RequirePackage{comment}
\RequirePackage{graphicx,subfigure,latexsym,amssymb}
\RequirePackage{float,epsfig,multirow,rotating,times}
\RequirePackage{upgreek,wrapfig}

\startlocaldefs

\newcommand {\ctp}{\citep}       

\numberwithin{equation}{section}
\theoremstyle{plain}
\newtheorem{theorem}{Theorem}[section]

\newtheorem{remark}{Remark}[section]

\newcommand{\bvarepsilon}{\boldsymbol{\varepsilon}}
\newcommand{\ba}{\boldsymbol{a}}
\newcommand{\bb}{\boldsymbol{b}}

\newcommand{\bof}{\boldsymbol{f}}
\newcommand{\boeta}{\boldsymbol{\eta}}

\newcommand{\balpha}{\boldsymbol{\alpha}}

\newcommand{\brho}{\boldsymbol{\rho}}

\newcommand{\bPhi}{\boldsymbol{\Phi}}

\newcommand{\bXi}{\boldsymbol{\Xi}}
\newcommand{\bxi}{\boldsymbol{\xi}}

\newcommand{\bV}{\boldsymbol{V}}

\newcommand{\bI}{\boldsymbol{I}}

\newcommand{\bS}{\boldsymbol{S}}

\newcommand{\bQ}{\boldsymbol{Q}}

\newcommand{\bW}{\boldsymbol{W}}
\newcommand{\bu}{\boldsymbol{u}}
\newcommand{\bv}{\boldsymbol{v}}
\newcommand{\bs}{\boldsymbol{s}}
\newcommand{\bU}{\boldsymbol{U}}
\newcommand{\bx}{\boldsymbol{x}}
\newcommand{\bX}{\boldsymbol{X}}
\newcommand{\by}{\boldsymbol{y}}
\newcommand{\bY}{\boldsymbol{Y}}

\newcommand{\btimes}{\boldsymbol{\times}}

\newcommand{\bnabla}{\boldsymbol{\nabla}}

\endlocaldefs

\begin{document}

\begin{frontmatter}
\title{Inverse Bayesian Estimation of Gravitational Mass Density in Galaxies from Missing Kinematic Data} 
\runtitle{Inverse Bayesian Estimation of Gravitational Mass Density in Galaxies}

\begin{aug}
\author{
{\fnms{Dalia} \snm{Chakrabarty}\thanksref{t1,t2,m1,m2}\ead[label=e1]{dc252@le.ac.uk}\ead[label=e2]{d.chakrabarty@warwick.ac.uk}},\fnms{Prasenjit} \snm{Saha}\thanksref{t3,m3}\ead[label=e3]{psaha@physik.uzh.ch}},
\thankstext{t1}{Lecturer in Statistics at Department of Mathematics, 
University of Leicester} 
\thankstext{t2}{Associate Research Fellow in Department of Statistics, 
University of Warwick} 
\thankstext{t3}{Lecturer in Institute for Theoretical Physics, 
University of Zurich} 

\runauthor{Chakrabarty $\&$ Saha}

\affiliation{University of Leicester\thanksmark{m1} \and University of Warwick\thanksmark{m2}, University of Zurich\thanksmark{m3}}

\address{Department of Mathematics\\
University of Leicester\\
Leicester LE1 3RH, 
U.K.\\
\printead*{e1}\\
\and\\
Department of Statistics\\
University of Warwick \\
Coventry CV4 7AL,
U.K.\\
\printead*{e2}\\
Institute for Theoretical Physics\\
University of Zurich\\
Winterthurerstr 190\\
8057 Zurich\\
Switzerland\\
\printead*{e3}
}
\end{aug}

\begin{abstract}
{ 
In this paper we focus on a type of inverse problem in which
the data is expressed as an unknown function of the sought and unknown
model function (or its discretised representation as a model parameter
vector). In particular, we deal with situations in which training data
is not available. Then we cannot model the unknown functional
relationship between data and the unknown model function (or
parameter vector) with a Gaussian Process of appropriate dimensionality. A
Bayesian method based on state space modelling is advanced
instead. Within this framework, the likelihood is expressed in terms
of the probability density function ($pdf$) of the state space
variable and the sought model parameter vector is embedded within the
domain of this $pdf$. As the measurable vector lives only inside an
identified sub-volume of the system state space, the $pdf$ of the
state space variable is projected onto the space of the measurables,
and it is in terms of the projected state space density that the
likelihood is written; the final form of the likelihood is achieved
after convolution with the distribution of measurement
errors. Application motivated vague priors are invoked and the
posterior probability density of the model parameter vectors, given
the data is computed. Inference is performed by taking posterior
samples with adaptive MCMC. The method is illustrated on synthetic as
well as real galactic data.\\[+2mm] {\bf Keywords}: Bayesian Inverse
Problems; State Space Modelling; Missing Data; Dark Matter in
Galaxies; Adaptive MCMC.}
\end{abstract}


\end{frontmatter}

\section{Introduction}
\label{sec:inv}
\noindent
The method of science calls for the understanding of
selected aspects of behaviour of a considered system, given
available measurements and other relevant information. The
measurements may be of the variable $\bW$ ($\bW\in{\cal
  W}\subseteq{\mathbb R}^m$) while the parameters that define the
selected system behaviour may be $\brho$, ($\brho\in{\cal
  B}\subseteq{\mathbb R}^d$) or the selected system behaviour can
itself be an unknown and sought function $\rho(\cdot)$ of the known
input variable vector $\bX$ ($\bX\in{\cal X}\subseteq{\mathbb R}^d$),
so that $\rho:{\cal X}\longrightarrow{\cal R}\subseteq{\mathbb R}$. In
either case, we relate the measurements with the model of the system
behaviour as in the equation $\bW = \bxi(\brho)$ or $\bW =
\bxi(\rho(\bX))$ where the function $\bxi(\cdot)$ is
unknown. Alternatively, in either case the scientist aims to solve an
inverse problem in which the operator $\bxi^{-1}$, when operated upon
the data, yields the unknown(s).

One problem that then immediately arises is the learning of the unknown function $\bxi(\cdot)$. Indeed $\bxi(\cdot)$ is often unknown though such is not the norm--for example in applications in which the data is generated by a known projection of the model function onto the space ${\cal W}$ of the measurables, $\bxi(\cdot)$ is identified as this known projection. Thus, image inversion is an example of an inverse problem in which the data is a known function of the unknown model function or model parameter vector \ctp[among others]{jugnon,qui2008,bereto,radon_xrays,bookblind}. On the other hand, there can arise a plethora of other situations in science in which a functional relationship between the measurable $\bW$ and unknown $\rho(\bX)$ (or $\brho$) is appreciated but the exact form of this functional relationship is not known \ctp[to cite a few]{parker,tarantola_siam,andrew,andrewreview,draper,tidal,gouveia}.

This situation allows for a (personal) classification 
of inverse problems such that
\begin{itemize}
\item in inverse problems of Type~I, $\bxi(\cdot)$ is known where $\bW = \bxi(\brho)$ or $\bW = \bxi(\rho(\bX))$,
\item in inverse problems of Type~II, $\bxi(\cdot)$ is unknown.
\end{itemize}
While inverse problems of Type~I can be rendered difficult owing to
these being ill-posed and/or ill-conditioned as well as in the
quantification of the uncertainties in the estimation of the
unknown(s), inverse problems of Type~II appear to be entirely
intractable in the current formulation of $\bW = \bxi(\brho)$ (or $\bW
= \bxi(\rho(\bX))$), where the aim is the learning of the unknown
$\brho$ (or $\rho(\bX)$), given the data. In fact, conventionally,
this very general scientific problem would not even be treated as an
inverse problem but rather as a modelling exercise specific to the
relevant scientific discipline. From the point of view of inverse
problems, these entail another layer of learning, namely, the learning
of $\bxi(\cdot)$ from the data--to be precise, from {\it training
  data} \ctp{training_geo,Caers,astro}. Here by training data we mean
data that comprises values of $\bW$ at chosen values of $\rho(\bx)$
(or at chosen $\brho$). These chosen (and therefore known) values of
$\rho(\bx)$ (or $\brho$) are referred to as the design points, so that
values of $\bW$ generated for the whole design set comprise the
training data.Having trained the model for $\bxi(\cdot)$ using such
training data, we then implement this learnt model on the available
measurements--or test data--to learn that value of $\brho$ (or
$\rho(\bx)$) at which the measurements are realised.

It is in principle possible to generate a training data set from
surveys (as in selected social science applications) or generate
synthetic training data sets using simulation models of the
system \ctp{simgeo, astrosim, atmos}. However, often the Physics of the situation is
such that $\bxi(\cdot)$ is rendered characteristic of the system
at hand (as in complex physical and biological systems). Consequently,
a simulation model of the considered system is only an approximation
of the true underlying Physics and therefore risky in general; after
all, the basic motivation behind the learning of the unknown
$\rho(\bx)$ (or $\brho$) is to learn the underlying system Physics, and
pivoting such learning on a simulation model that is of unquantifiable
crudeness, may not be useful.

Thus, in such cases, we need to develop an alternative way of learning
$\bxi(\cdot)$ or if possible, learn the unknown $\rho(\bx)$ (or
$\brho$) given the available measurements without needing to know
$\bxi(\cdot)$. It may appear that such is possible in the Bayesian
approach in which we only need to write the posterior probability
density of the unknown $\rho(\bx)$ (or $\brho$), given the data. An
added advantage of using the Bayesian framework is that extra
information is brought into the model via the priors, thus reducing
the quantity of data required to achieve inference of a given
quality. Importantly, in this approach one can readily achieve
estimation of uncertainties in the relevant parameters, as
distinguished from point estimates of the same. In this paper we
present the Bayesian learning of the unknown model parameters given
the measurements but no training data, as no training data set is
available. The presented methodology is inspired by state space
modelling techniques and is elucidated using an application to
astronomical data.

The advantages of the Bayesian framework notwithstanding, in systems
in which training data is unavailable, fact remains that $\bxi(\cdot)$
cannot be learnt. This implies that if learning of the unknown
$\rho(\bx)$ (or $\brho$) is attempted by modelling $\bxi(\cdot)$ as a
realisation from a stochastic process (such as a Gaussian Process
(${\cal GP}$) or Ito Process or t-process, etc.), then the correlation
structure that underlies this process is not known. However, in this
learning approach, the posterior probability of the unknowns given the
data invokes such a correlation structure. Only by using training data
can we learn the covariance of the process that $\bxi(\cdot)$
is sampled from, leading to our formulation of the posterior of the
unknowns, given the measured data as well as the training data. To
take the example of modelling $\bxi(\cdot)$ using a high-dimensional
${\cal GP}$, it might be possible of course to impose the form of the
covariance by hand; for example, when it is safe to assume that
$\bxi(\cdot)$ is continuous, we could choose a stationary covariance
function \ctp{rasmussen}, such as the popular square exponential
covariance or the Matern class of covariance functions
\ctp{materngneiting}, though parameters of such a covariance
(correlation length, smoothness parameter) being unknown, ${\it
  ad\:\:hoc}$ values of these will then need to be imposed. In the
presence of training data, the smoothness parameters can be learnt
from the data.

For systems in which the continuous assumption is misplaced, choosing
an appropriate covariance function and learning the relevant
parameters from the measured data, in absence of training data,
becomes even trickier. An example of this situation can arise in fact
in an inverse problem of Type~I--the unknown physical density of the
system is projected onto the space of observables such that inversion
of the available (noisy) image data will allow for the estimation of
the unknown density, where the projection operator is known. Such a
density function in real systems can often have disjoint support in
its domain and can also be typically characterised by sharp density
contrasts as in material density function of real-life material
samples \ctp{CRP}. Then, if we were to model this discontinuous and
multimodal density function as a realisation from a ${\cal GP}$, the
covariance function of such a process will need to be
non-stationary. It is possible to render a density function sampled
from such a ${\cal GP}$ to be differently differentiable at different
points, using for example prescriptions advanced in the literature
\ctp{paciorek}, but in lieu of training data it is not possible to
parametrise covariance kernels to ensure the representative
discontinuity and multimodality of the sampled (density)
functions. Thus, the absence of training data leads to the inability
to learn the correlation structure of the density function given the
measured image data.

A way out this problem could be to make an attempt to construct a
training data set by learning values of the unknown system behaviour
function at those points in the domain of the density, at which
measured data are available; effectively, we then have a set of data
points, each generated at a learnt value of the function, i.e. this
set comprises a training data. In this data set there are measurement
uncertainties as well as uncertainty of estimation on each of the
learnt values of the system function. Of course, learning the value of
the function at identified points within the domain of the system
function, is in itself a difficult task. Thus, in this paradigm, the
domain ${\cal X}\subseteq{\mathbb R}^d$ of the unknown system function
$\rho(\bx)$ is discretised according to the set of values of $\bX$,
$\{\bx_1,\bx_2,\ldots,\bx_n\}$, at which the $n$ measurements are
available. In other words, the discretisation of ${\cal X}$ is dictated
by the data distribution. Over each $\bX$-bin, the function
$\rho(\bx)$ is held a constant such that for $\bX$ in the $i$-th bin,
the function takes the value $\rho_i$, $i=1,2,\ldots,n$; then we
define $\brho:=(\rho_1,\rho_2,\ldots,\rho_n)^T$ and try to learn this
vector, given the data. Unless otherwise motivated, in general
applications, the probability distribution of $\rho_i$ is not imposed
by hand. In the Bayesian framework this exercise translates to the
computing of the joint posterior probability density of $n$
distribution-free parameters $\rho_1,\rho_2,\ldots,\rho_n$ given the
data, where the correlation between $\rho_i$ and $\rho_j$ is not
invoked, $i,j=1,2,\ldots,n;\:i\neq j$. Of course, framed this way, we
can only estimate the value of the sought function $\rho(\bx)$ at
identified values of $\bX$--unless interpolation is used--but once the
training data, thus constructed, is subsequently implemented in the
modelling of $\bxi(\cdot)$ with a ${\cal GP}$ of appropriate
dimensionality, statistical prediction at any value of $\bX$ may be
possible.

Above, we dealt schematically with the difficult case of lack of
training data.  However, even when a training data set is available,
learning $\bxi(\cdot)$ using such data can be hard. In principle,
$\bxi(\cdot)$ can be learnt using splines or wavelets. However, a
fundamental shortcoming of this method is that splines and wavelets
can fail to capture the correlation amongst the component functions of
a high-dimensional $\bxi(\cdot)$. Also, the numerical difficulty of
the very task of learning $\bxi(\cdot)$ using this technique, and
particularly of inverting the learnt $\bxi(\cdot)$, only increases with
dimensionality. Thus it is an improvement to model such a
$\bxi(\cdot)$ with a high-dimensional ${\cal GP}$. A high-dimensional
$\bxi(\cdot)$ can arise in a real-life inverse problem if the observed
data is high-dimensional, eg. the data is matrix-variate \ctp{CBB}.

Measurement uncertainties or measurement noise is almost unavoidable in practical applications and therefore, any attempt at an inference on the unknown model parameter vector $\brho$ (or the unknown model function $\rho(\bx)$) should be capable of folding in such noise in the data. In addition to this, there could be other worries stemming from inadequacies of the available measurements--the data could be ``too small" to allow for any meaningful inference on the unknown(s) or ``too big" to allow for processing within practical time frames; here the qualification of the size of the data is determined by the intended application as well as the constraints on the available computational resources. However, a general statement that is relevant here is the fact that in the Bayesian paradigm, less data is usually required than in the frequentists' approach, as motivated above. Lastly, data could also be missing; in particular, in this paper we discuss a case in which the measurable lives in a space ${\cal U}\subset{\cal W}$ where ${\cal W}$ is the state space of the system at hand.

The paper is constructed as follows. In Section~\ref{sec:intro}, we
briefly discuss the outline of state space modelling. In the following
Section~\ref{sec:generic}, our new state space modelling based
methodology is delineated; in particular, we explore alternatives to
the suggested method in subsection~\ref{sec:alternative}. The
astrophysical background to the application using which our
methodology is elucidated, is motivated in Section~\ref{sec:casestudy}
while the details of the modelling are presented in
Section~\ref{sec:modelreal}. We present details of our inference in
Section~\ref{sec:inference} and applications to synthetic and real
data are considered in Section~\ref{sec:synthetic} and
Section~\ref{sec:real} respectively. We round up the paper with some
discussions about the ramifications of our results in
Section~\ref{sec:discussions}.

\section{State Space Modelling}
\label{sec:intro}
\noindent
Understanding the evolution of the probability density function of the
state space of a dynamical system, given the available data, is of
broad interest to practitioners across disciplines. Estimation of the
parameters that affect such evolution can be performed within the
framework of state space models or SSMs
\ctp{West97,polewesthar,harveybook,Carlin92}. Basically, an SSM
comprises an observation structure and an evolution
structure. Assuming the observations to be conditionally independent,
the marginal distribution of any observation is dependent on a known
or unknown stationary model parameter, at a given value of the state
space parameter at the current time. Modelling of errors of such
observations within the SSM framework is of interest in different
disciplines \ctp{winshipecology,birdstate}.

The evolution of the state space parameter is on the other hand given
by another set of equations, in which the uncertainty of the evolved
value of the parameter is acknowledged. 
A state space representation 
of complex systems will in general have to be designed to capacitate
high-dimensional inference in which both the
evolutionary as well as observation equations are in general non-linear and
parameters and uncertainties are non-Gaussian.

In this paper we present a new methodology that offers a state space
representation in a situation when data is collected at only one time
point and the unknown state space parameter in this treatment is
replaced by the discretised version of the multivariate probability
density function ($pdf$) of the state space variable. The focus is on the
learning of the static unknown model parameter vector rather than on
prediction of the state space parameter at a time point different to
when the observations are made. In fact, the sought model parameter
vector is treated as embedded within the definition of the $pdf$ of the state
space variable. In particular, the method that we present here pertains to a
partially observed state space, i.e. the observations comprise
measurements on only some--but not all--of the components of the state
space vector. Thus in this paradigm, probability of the observations
conditional on the state space parameters reduces to the probability
that the observed state space data have been sampled from
the $pdf$ of the full state space variable vector, marginalised over the unobserved components. Here this $pdf$ includes the sought static model
parameter vector in its definition. 
In addition to addressing missing data, the presented methodology is
developed to acknowledge the measurement errors that may be
non-Gaussian.

The presented method is applied to real and synthetic astronomical
data with the aim of drawing inference on the distribution of the
gravitational mass of all matter in a real and simulated galaxy,
respectively. This gravitational mass density is projected to be
useful in estimating the distribution of dark matter in the galactic
system.

\section{Method in general}
\label{sec:generic}
\noindent
Here we aim to learn the unknown model parameter vector $\brho$ given
the data, where data comprises $N_{data}$ measurements of some
($h$) components of the $d$-dimensional state space parameter vector
$\bX$; thus, $h < d$. Here $\bX=(X_1,X_2,\ldots,X_d)^T$. 
In fact, the data set is
$\{\bu^{(i)}\}_{i=1}^{N_{data}}$ where the $i$-th observation is the
vector $\bU^{(i)}=(X_1^{(i)}, X_2^{(i)}, \ldots, X_h^{(i)})^T$. Let
the state space be ${\cal W}$ so that
$\bX\in{\cal W}$. Let the observable vector be
$\bU\in\cal{U}\subset{\cal W}$. Let
$\Pr(\bX\in[\bx,\bx+d\bx])=f_{\bX}(\bx, \balpha)d\bx$, i.e. the
probability density function of the state parameter vector $\bX$ is
$f(\bx, \balpha)$, where the distribution is parametrised by the
parameter $\balpha\in{\mathbb R}^j$.

In light of this, we suggest that $\bX \sim f(\bx,\balpha)$.
Then had the observations lived in the state
space ${\cal W}$, we could have advanced the likelihood function in terms of
$f(\cdot,\cdot)$. However, here we deal with missing data that we know
lives in the sub-space ${\cal U}$ within ${\cal
  W}$. Therefore, the data must be sampled from the density
$\nu(\bu,\balpha)$ that is obtained by marginalising the $pdf$
$f(\bx,\balpha)$ over $X_{h+1}, X_{h+2}, \ldots, X_d$. In other
words, the $pdf$ $f(\bx,\balpha)$ is projected onto the space of the
observables, i.e. onto ${\cal U}$; the result is the projected or
marginalised density $\nu(\bu,\balpha)$ of the observables. Then under
the assumption of the observed vectors being conditionally $iid$, the
likelihood function is
\begin{equation}
\Pr(\{\bu^{(i)}\}_{i=1}^{N_{data}}\vert\balpha) = \displaystyle{\prod_{i=1}^{N_{data}} \nu(\bu^{(i)},\balpha)}
\label{eqn:likeli_prelim}
\end{equation}
where\\
\begin{equation}
\nu(\bu^{(i)},\balpha)=
{\displaystyle{\int\limits_{X_{h+1}}\ldots\int\limits_{X_d} f(x_1^{(i)},\ldots,x_h^{(i)},x_{h+1},\ldots,x_d,\balpha) dx_{h+1}\ldots dx_{d}}}.
\end{equation}

While the likelihood is thus defined, what this definition still does
not include in it is the sought model parameter vector $\brho$. In
this treatment, we invoke a secondary equation that allows for the
model parameter vector $\brho$ to be embedded into the definition of
the likelihood. This can be accomplished by eliciting
application specific details but in general, we suggest
$\balpha=\bxi(\brho):=(\xi_1(\brho),\xi_2(\brho),\ldots\xi_j(\brho))^T$
and construct the general model for the state space $pdf$ to be
\begin{equation}
f(\bx,\balpha)\equiv f(\boeta(\bx),\bxi(\brho))
\label{eqn:f_fin}
\end{equation}
where $\boeta(\cdot)$ is a $t$-dimensional vector function of a vector.

Given this rephrasing of the state space $pdf$, the projected density that the $i$-th measurement $\bu^{(i)}$ is sampled from, is re-written as\\
\begin{equation}
\nu(\bu^{(i)},\brho) = 
{
\displaystyle{\int\limits_{X_{h+1}}\ldots\int\limits_{X_d} 
f(\boeta(x_1^{(i)},\ldots,x_h^{(i)},x_{h+1},\ldots,x_d),\bxi(\brho)) dx_{h+1} \ldots dx_{d}}}
\label{eqn:nu_second}
\end{equation}
so that plugging this in the RHS of Equation~\ref{eqn:likeli_prelim}, the likelihood is 
\begin{equation}
\Pr(\{\bu^{(i)}\}_{i=1}^{N_{data}}\vert\brho) = \displaystyle{\prod_{i=1}^{N_{data}} \nu(\bu^{(i)},\brho)}
\label{eqn:likeli_second}
\end{equation}

However, it is appreciated that the $pdf$ of the state space vector
$\bX$ may not be known, i.e. $f(\cdot,\cdot)$ is unknown. This
motivates us to attempt to learn the state space $pdf$ from the data,
simultaneously with $\brho$. We consider the situation that training
data is unavailable where training data would comprise a set of values
of $\bU$ generated at chosen values of $\rho(\bX)$. However, since the
very functional relationship ($\bxi(\cdot)$ in the notation motivated
above) between $\bU$ and $\rho(\bX)$ is not known, it is not possible
to generate values of $\bU$ at a chosen value of $\rho(\bX)$, unless
of course, an approximation of unquantifiable crudeness for this
functional relationship is invoked. Here we attempt to improve upon the
prospect of imposing an ${\it ad\:\:hoc}$ model of $\bxi(\cdot)$.  Then
in this paradigm, we discretise the function
$f(\boeta(\bx),\bxi(\brho))$. 

This is done by placing the relevant ranges of the vectors
$\boeta(\bx)$ and $\bxi(\brho)$ on a grid of chosen cell size. Thus,
for $\boeta(\cdot)$ and $\bxi(\cdot)$ being discretised into $t$ and
$j$-dimensional vectors respectively, the discretised version of
$f(\boeta(\bx),\bxi(\brho))$ is then represented as the $t\times
j$-dimensional vector $\bof$ such that the $p$-th component of this
vector is the value of $f(\boeta(\bx),\bxi(\brho))$ in the $p$-th
``$\boeta-\bxi$-grid cell''.  Here, such a grid-cell is the $p$-th of
the ones that the domain of $f(\cdot,\cdot)$ is discretised into,
$p=1,2,\ldots,p_{max}$. 

Given this discretisation of $f(\cdot,\cdot)$, the RHS of Equation~\ref{eqn:nu_second} is reduced to a sum of integrals over the unobserved variable in each of the grid-cells. In other words,
\begin{equation}
\nu(\bu^{(i)},\brho,\bof) = \displaystyle{
\sum_{p=1}^{p_{max}} \left[f_p
\displaystyle{\int_{\by^{(p-1)}(\bu^{(i)},\:\brho)}^{\by^{(p)}(\bu^{(i)},\:\brho)}
d\by^{'}}\right]}
\label{eqn:nu_third}
\end{equation}
where $\by^{(p)}(\bu^{(i)},\:\brho)$ is the value that the vector of
the unobserved variables takes up in the $p$-th
$\boeta-\bxi$-grid-cell. The integral on the RHS of
Equation~\ref{eqn:nu_third} represents the volume that the $p$-th
$\boeta-\bxi$-grid-cell occupies in the space of the unobserved variable vector
$\bY=(X_{h+1},X_{h+2},\ldots,X_d)^T$. The
value of $\bY$ in the $p$-th $\boeta-\bxi$-grid-cell is dependent in
general on $\brho$ for a given data vector $\bu^{(i)}$; hence the
notation $\by^{(p)}(\bu^{(i)},\:\brho)$.

In other words, to compute the integral for each $p$ (on the RHS of
Equation~\ref{eqn:nu_third}) we need to identify the bounds on the
value of each component of $\bY$ imposed by the edges of the $p$-th
$\boeta-\bxi$ grid-cell. This effectively calls for identification of
the mapping between the space of $\boeta(\bx)$ and $\bxi(\brho)$, and
the space of the unobserved variables $\bY$. Now the observation
$\bU\in{\cal U}\subset{\cal W}$. Then $\bY\in{\cal Y}$, where ${\cal
  Y}\oplus{\cal U}={\cal W}$. Indeed, this mapping will be understood
using the physics of the system at hand. We will address this
in detail in the context of the application that is considered in the
paper.

The likelihood function is then again rephrased as
\begin{eqnarray}
&&\Pr(\{\bu^{(i)}\}_{i=1}^{N_{data}}\vert\brho,\bof) = \displaystyle{\prod_{i=1}^{N_{data}} \nu(\bu^{(i)},\brho,\bof)}\nonumber \\
\label{eqn:likeli_fin}
&&=\displaystyle{\prod_{i=1}^{N_{data}}\displaystyle{
\sum_{p=1}^{p_{max}} \left[f_p
\int_{\by^{(p-1)}(\bu^{(i)},\:\brho)}^{\by^{(p)}(\bu^{(i)},\:\brho)}
d\by^{'}\right]}}
\end{eqnarray}
using Equation~\ref{eqn:nu_third}.

However, the observed data is likely to be noisy too. To incorporate
the errors of measurement, the likelihood is refined by convolving
$\nu(\bu^{(i)},\brho,\bof)$ with the density of the error
$\bvarepsilon$ in the value of the observed vector $\bU$, where the
error distribution is assumed known. Let the density of the error
distribution be $g(\bu; \bvarepsilon)$ where $\bvarepsilon$ are the
known parameters. Then the likelihood is finally advanced as
\begin{equation}
\Pr(\{\bu^{(i)}\}_{i=1}^{N_{data}}\vert\brho,\bof) = \displaystyle{\prod_{i=1}^{N_{data}} \nu(\bu^{(i)},\brho,\bof)\ast g(\bu^{(i)};\bvarepsilon)}
\label{eqn:likeli_fin2}
\end{equation}

In a Bayesian framework, inference is pursued thereafter by selecting
priors for the unknowns $\brho$ and $\bof$, and then using the
selected priors in conjunction with the likelihood defined in
Equation~\ref{eqn:likeli_fin2}, in Bayes rule to give the posterior of
the unknowns given the data, i.e $\pi(\brho,\bof\vert
\{\bu^{(i)}\}_{i=1}^{N_{data}})$. In the context of the application at
hand, we will discuss all this and in particular, advance the
data-driven choice of the details of the discretisation of the
$f(\boeta(\bx),\bxi(\brho))$ function. Posterior samples could be
generated using a suitable version of Metropolis-Hastings and
implemented to compute the 95$\%$ HPD credible regions on the learnt
parameter values.

\subsection{Alternative methods}
\label{sec:alternative}
\noindent
We ask ourselves the question about alternative treatment of the data
that could result in the estimation of the unknown model parameter
vector $\brho$. Let the sought model parameter be $s$-dimensional
while the observable $\bU$ is an $h$-dimensional vector valued variable and
there are $N_{data}$ number of measurements of this variable
available. Then the pursuit of
$\rho$ can lead us to express the data as a function of the model
parameter vector, i.e. write ${\bU}=\bXi(\brho)$, where
$\bXi(\cdot)$ is an unknown, $h$-dimensional vector valued
function of an $s$-dimensional vector. In order to learn $\brho$, we
will need to first learn $\bXi(\cdot)$ from the data, as was
motivated in the introductory section.

As we saw in that section, the learning of this high-dimensional
function from the data and its inversion are best tackled by modelling
the unknown high-dimensional function with a Gaussian
Process. \ctp{CBB} present a generic Bayesian method that
performs the learning and inversion of a high-dimensional function
given matrix-variate data within a supervised learning paradigm; the
(chosen) stationary covariance function implemented in this work is
learnt using training data and is subsequently used in the computation
of the posterior probability of the unknown model parameter vector
given the measured or test data, as well as the training data. In the
absence of available training data, such an implementation is not
possible, i.e. such a method is not viable in the unsupervised learning
paradigm. In the application we discuss below, training data is not
available and therefore, the modelling of the functional relation
between data and $\brho$, using Gaussian Processes appears to not be
possible. This shortcoming can however be addressed if simulations of
the system at hand can be undertaken to yield data at chosen values of
$\brho$; however, the very physical mechanism that connects $\brho$
with the data may be unknown (as in the considered application) and
therefore, such a simulation model is missing.  Alternatively, if
independently learnt $\brho$, learnt with an independent data set, is
available, the same can be used as training data to learn $\brho$ given
another data set. On such instances, the Gaussian Process approach is
possible but in lieu of such training data becoming available, the
learning of $\brho$ given the matrix-valued data can be performed in
the method presented above. On the other hand, a distinct advantage
of the method presented below is that it allows for the learning of
the state space density in addition to the unknown model parameter
vector.

If the suggestion is to learn the unknown system function $\rho(\bX)$
as itself a realisation of a ${\cal GP}$, the question that then needs
to be addressed is how to parametrise the covariance structure of
${\cal GP}$ in situations in which the data results from measurements
of the variable $\bU$ that shares an unknown functional relation with
$\rho(\bX)$. In other words, in such situations, the unknown system
function $\rho(\bX)$ has to be linked with the available data via a
functional relation, which however is unknown, as motivated above; we
are then back to the discussion in the previous paragraph.

\section{Case study}
\label{sec:casestudy}
\noindent
Unravelling the nature of Dark Matter and Dark Energy is one of the
major challenges of today's science. While such is pursued, the
gathering of empirical evidence for/against Dark Matter (DM) in
individual real-life observed astronomical systems is a related
interesting exercise.

The fundamental problem in the quantification of dark matter in these
systems is that direct observational evidence of DM remains elusive.
In light of this, the quantification is pursued using information
obtained from measurable physical manifestations of the gravitational
field of all matter in an astronomical system, i.e. dark as well as
self-luminous matter. Indeed, such measurements are difficult and
physical properties that manifest the gravitational effect of the
total gravitational field of the system would include the density of
X-rays emitted by the hot gas in the system at a measured temperature
\ctp{3379_xray}, velocities of individual
particles that live in the system and play in its gravitational field
\ctp{coccato09, cote_m49, romanowskyscience, m15,chakrabarty_somak} and the deviation in the path of a ray of light brought
about by the gravitational field of the system acting as a
gravitational lens \ctp{koopmans_06}. 

The extraction of the density of DM from the learnt total
gravitational mass density of all matter in the system, is performed
by subtracting from the latter, the gravitational mass density of the
self-luminous matter. The density of such luminous matter is typically
modelled astronomically using measurements of the light that is
observed from the system. A reliable functional relationship
between the total gravitational mass density and such photometric
measurements is not motivated by any physical theories though
the literature includes such a relationship as obtained
from a pattern recognition study performed with a chosen class of
galaxies \ctp{brendan}. 

In this work, we focus our attention to the learning of the total
gravitational mass density in galaxies, the images of which resemble
ellipses - as distinguished from disc-shaped galaxies for which the
sought density is more easily learnt using measurement of rotational
speed of resident particles. By a galactic ``particle'' we refer to
resolved galactic objects such as stars. There could also be
additional types of particles, such as planetary nebulae (PNe) which
are an end state of certain kinds of stars; these bear signature
marks in the emitted spectral data. Other examples of galactic particles could
include old clusters of stars, referred to as globular clusters (GCs).


\subsection{Data}
\label{sec:backgnd}
\noindent
As defined above, the space of all states that a dynamical system
achieves is referred to as the system's state space ${\cal W}$. Now, the state that a galaxy is
in, is given by the location and velocity coordinates of all particles
in the system. Here, the location coordinate is $\bX\in{\mathbb R}^3$
as is the velocity coordinate vector $\bV$. Thus, in our treatment of
the galaxy at hand, ${\cal W}$ is the space of the particle location
and velocity vector i.e. the space of the vector $\bW=(\bX^T,\bV^T)^T$. We
model the galactic particles to be playing in the average
(gravitational) force field that is given rise to by all the particles
in this system. Under the influence of this mean field, we assume the
system to have relaxed to a stationary state so that there is no time
dependence in the distribution of the vector $\bW=(\bX^T,\bV^T)^T$,
where the 3-dimensional vector $\bX=(X_1, X_2, X_3)^T$ and
$\bV=(V_1,V_2,V_3)^T$. Then the $pdf$ of the variable
$(\bX^T,\bV^T)^T$ is
$f(X_1,X_2,X_3,V_1,V_2,V_3,\balpha)$, where $\balpha$ is a
parameter vector.

Our aim is to learn the density function of gravitational mass of all
matter in the galaxy, given the data ${\bf
  D}=\{\bu_i\}_{i=1}^{N_{data}}$, where $\bU=(X_1,X_2,V_3)^T$. The
physical interpretation of these observables is that $V_3$ is the
component of the velocity of a galactic particle that is aligned along
the line-of-sight that joins the particle and the observer, i.e. we
can measure how quickly the particle is coming towards the observer or
receding away but cannot measure any of the other components of $\bV$.
Similarly, we know the components $X_1$ and $X_2$ of the location
$\bX$ of a galactic particle in the galactic image but cannot observe
how far orthogonal to the image plane the particle is, i.e. $X_3$ is
unobservable. Thus $\bU=(X_1,X_2,V_3)^T\in{\cal U}$ but
$\bW:=(X_1,X_2,X_3,V_1,V_2,V_3)^T\in{\cal W}$ with ${\cal
  U}\subset{\cal W}$.  It merits mention that in the available data,
values of $X_1$ and $X_2$ appear in the form of $\sqrt{x_1^2+x_2^2}$.
Then the data ${\bf
  D}=\{\bu_i\}_{i=1}^{N_{data}}\equiv\left\{\sqrt{(x_1^{(k)})^2+(x_2^{(k)})^2},v_3^{(k)}\right\}_{k=1}^{N_{data}}$.

Here $N_{data}$ is typically of the order of 10$^2$. While for
more distant galaxies, $N_{data}$ is lower, recent advancements is
astronomical instrumentation allows for measurement of $V_3$ of around
750 planetary nebulae or PNe (as in the galaxy CenA, Woodley,
$\&$ Chakrabarty, under preparation). Such high a sample size is
however more of an exception than the rule - in fact, in the real
application discussed below, the number of $V_3$ measurements of
globular clusters (or GCs) available is only 29.  In addition, the
measurements of $V_3$ are typically highly noisy, the data would
typically sample the sub-space ${\cal U}$ very sparsely and the data
sets are typically one-time measurements. The proposed method will
have to take this on board and incorporate the errors in the
measurement of $V_3$. Given such data, we aim to learn the
gravitational mass density of all matter - dark
as well as self-luminous - at any location $\bX$ in the galaxy. 

\section{Modelling real data}
\label{sec:modelreal}
\noindent
In the Bayesian framework, we are essentially attempting to compute
the posterior of the unknown gravitational mass density function
$\rho(\bX)$, given data ${\bf D}$.  Since
gravitational mass density is non-negative, $\rho:{\mathbb
  R}^3\longrightarrow{\mathbb R}_{\geq 0}$. That we model the mass
density to depend only on location $\bX$ is a model
assumption\footnotemark.  \footnotetext{We assume that (the system is
  Hamiltonian so that) the gravitational potential of the galaxy is
  independent of velocities and depend only on location; since
  gravitational potential is uniquely determined for a given system
  geometry, by the gravitational mass density (via Poisson Equation),
  the latter too is dependent only on $\bX$.}

From Bayes rule, the posterior probability density of $\rho(\bX)$
given data ${\bf D}$ is given as proportional to the product of the
prior and the likelihood function, i.e. the probability density of
${\bf D}$ given the model for the unknown mass density. Now, the
probability density of the data vector $\bU$ given the model
parameters $\balpha$ is given by the probability density function
$\nu(\bU,\balpha)$ of the observable $\bU$, so that, assuming the
$N_{data}$ data vectors to be conditionally independent, the
likelihood function is the product of the $pdf$s of $\bU$ obtained at the
$N_{data}$ values of $\bU$:
\begin{equation}
{\cal L}(\balpha\vert {\bf D})= \displaystyle{\prod_{k=1}^{N_{data}} \nu(x_1^{(k)},x_2^{(k)},v_3^{(k)},\balpha)}.
\end{equation}
This is Equation~\ref{eqn:likeli_fin} written in the context of this
application.  Given that $\bU\in{\cal U}\subset{\cal W}$, the $pdf$ of
$\bU$ is related to the $pdf$ $f(\bX,\bV,\balpha)$ of the
vector-valued variable $\bW\equiv(\bX^T,\bV^T)^T$ as\\
\begin{equation}
\nu(x_1,x_2,v_3,\balpha) = 
\displaystyle{\int_{X_3}\int_{V_1}\int_{V_2} f(x_1,x_2,x_3,v_1,v_2,v_3,\balpha)dx_3 dv_1 dv_2}.
\end{equation}
However, this formulation still does not include the gravitational
mass density function $\rho(\bx)$ in the definition of
$f(\bX,\bV)$, we explore the Physics of the situation to find how to
embed $\rho(\bx)$ into the definition of the $pdf$
of the state space variable $\bW$, and thereby into the
likelihood. This is achieved by examining the time evolution of this $pdf$
of the state space variable; we discuss this next.

\subsection{Evolution of $f(\bX,\bW)$ and embedding $\rho(\bX)$ in it}
\noindent
Here we invoke the secondary equation that tells of the evolution of
$f(\bX,\bV)$. In general, the $pdf$ of the state space variable is a function of
$\bX$, $\bV$ and time $T$. So the general state space $pdf$ is
expected to be written as $f(\bX,\bV,T)$, with $f:{\cal W}\times{\cal
  T}\longrightarrow{\mathbb R}_{\geq 0}$. It is interpreted as the
following: at time $t$ ($t\in{\cal T}$), the probability for
$\bX\in[\bx,\bx+d\bx]$ and $\bV\in[\bv,\bv+d\bv]$ for a galactic
particle is $f(\bx,\bv,t)d^3 \bx d^3 \bv$. However, we assume that the
particles in a galaxy do not collide since the galactic particles
inside it, (like stars), typically collide over time-scales that are
$\gtrsim$ the age of galaxies \ctp{BT}. Given this assumption of
collisionlessness, the $pdf$ of $\bW=(\bX^T,\bV^T)^T$ remains
invariant. Thus, the
evolution of $f({\bx},\bv,t)$ must is guided by the Collisionless
Boltzmann Equation (CBE):
\begin{equation}
\displaystyle{\frac{df}{dt}} = \displaystyle{\frac{\partial f}{\partial t} + \sum_{i=1}^{3} {\dot{x}}_i\frac{\partial f}{\partial x_i} + \sum_{i=1}^{3} {\dot{v}}_i\frac{\partial f}{\partial v_i}} = 0.
\label{eqn:cbe}
\end{equation}
This equation suggests that when the state space 
distribution has attained stationarity, so that
$\displaystyle{\frac{\partial f}{\partial t}}=0$, $f(\bx,\bv)$ is a
constant $\forall\:\:\bx,\:\bv$ at a given time. This is
referred to as Jeans theorem \ctp{BT}. {In fact, the
  equation more correctly suggests that as long as the system has
  reached stationarity, at any given time, $f(\bx,\bv)$ is a constant
  $\forall\:\:\bx,\:\bv$ inside a well-connected region$\subseteq{\cal W}$.}  
Given this, the state space
$pdf$ can be written as a function of quantities that do not change
with time\footnotemark. \footnotetext{To be precise, the state space
  $pdf$ should be written as a function of integrals of motion, which
  remain constant along the trajectory from one point in ${\cal W}$ to
  another, during the motion.}

\begin{theorem}
\label{th:cbe}
\noindent
Any function $I(\bx,\bv)$ is a steady-state or stationary solution of
the Collisionless Boltzmann Equation i.e. a solution to the equation
$\displaystyle{\frac{df}{dt}}$=0 if and only if
$I(\bx,\bv)$ is invariant with respect to time, for all $\bx$ and $\bv$
that lie inside a well-connected region$\subseteq{\cal W}$.
\end{theorem}

\begin{proof}
The proof is simple; for the proof we assume $\bX$ and $\bV$ to take respective values of $\bx$ and $\bv$ inside a well-connected sub-space of ${\cal W}$. 
Let a function of the vectors $\bx$, $\bv$ be $I(\bx,\bv)$ such that it 
remains a constant w.r.t. time. Then
$\displaystyle{\frac{dI(\bx,\bv)}{dt}}=0\Longrightarrow$this
function is a solution to the equation
$\displaystyle{\frac{df}{dt}}$=0.
 
Let the equation $\displaystyle{\frac{df}{dt}}$=0 have a solution
$J(\bx,\bv,t)$. This implies $\displaystyle{\frac{dJ(\bx,\bv,t)}{dt}}=0$,
i.e. $J(\bx,\bv,t)$ is a constant with respect to time. For this to be true, $J(\bx,\bv,t)\equiv I(\bx,\bv)$. Therefore the solution to
$\displaystyle{\frac{df}{dt}}$=0 is a function of $\bx$ and $\bv$ that
is a constant w.r.t. time.
\end{proof}
In fact, any function of a time-invariant function of vectors $\bX$
and $\bV$ is also a solution to the CBE.

Now, in our work we assume the system to have attained stationarity so
that the $pdf$ of the state space variable has no time
dependence. Then the above theorem suggests that we can write
$f(\bx,\bv)=g(I_1(\bx,\bv), I_2(\bx,\bv), \ldots, I_n(\bx,\bv))$ for
any $n\in{\mathbb Z}_{+}$, where $I_i(\cdot,\cdot)$ is any
time-independent function of 2 vectors, for $i=1,2,\ldots,n$.

Now, upon eliciting from the literature in galactic dynamics \ctp{contop63, binney82} we realise the
following.
\begin{itemize}
\item The number $n$ of constants of motion can be at most 5, i.e. $n=1,2,3,4,5$.
\item The $pdf $ of the state space variable has to include particle energy $E(\bX,\bV)$, (which is one constant of motion), in its domain. Thus, we can write $f(\bX,\bV)=f(E(\bX,\bV), I_2(\bX,\bV), \ldots, I_5(\bX,\bV))$.
\item Energy $E(\bX,\bV)$ is given as the sum of potential energy
  $\bPhi(\parallel\bX\parallel)$ and kinetic energy $\bV\cdot\bV/2$, i.e.
\begin{eqnarray}
&&E(\bX,\bV)=\bPhi(\parallel\bX\parallel)+\bV\cdot\bV/2,\\ 
&&\parallel\bX\parallel\equiv\sqrt{\bX\cdot\bX}\equiv\sqrt{X_1^2+X_2^2+X_3^2},\\
&&\mbox{with}\quad\bX=(X_1,X_2,X_3)^T,\nonumber\\ 
&&\bV\cdot\bV\equiv\parallel\bV\parallel^2 \equiv V_1^2+V_2^2+V_3^2,\\
&&\mbox{with}\quad\bV=(V_1,V_2,V_3)^T\nonumber.
\end{eqnarray}
Here $\parallel\cdot\parallel$ is the Euclidean norm. That the potential is maintained as dependent only of the location vector $\bX$ and not on $\bV$ stems from our assumption that there is no dissipation of energy in this system, i.e. we model the galaxy at hand to be a Hamiltonian system. Here, a basic equation of Physics relates the potential of the galaxy to the gravitational mass density of the system, namely Poisson Equation:
\begin{eqnarray}
\bnabla^2\Phi(R) = -4\pi G\rho(R)\quad\mbox{where}&&\\
\label{eqn:poisson}
R := \parallel\bX\parallel= \sqrt{X_1^2+X_2^2+X_3^2}&&,
\end{eqnarray}
$\bnabla^2$ is the Laplace operator (in the considered geometry of the galaxy) and $G$ is a known constant (the Universal gravitational constant).  
\end{itemize}
On the basis of the above, we can write
\begin{eqnarray}
&&f(\bX,\bV)\nonumber \\
&=& f(E(\bX,\bV), I_2(\bX,\bV), \ldots, I_5(\bX,\bV))\\ \nonumber
&=& f(E(\Phi(\rho(R)),\bV\cdot\bV), I_2(\bX,\bV), \ldots, I_5(\bX,\bV))\\ \nonumber
&=& f(E(\bX\cdot\bX,\bV\cdot\bV), I_2(\bX,\bV), \ldots, I_5(\bX,\bV))\\ 
\label{eqn:fmanyforms}
\end{eqnarray}

At this point we recall the form of an isotropic function of 2 vectors \ctp{shi, truesdell, wang}.
\begin{remark}
\label{remark:isotropic}
\noindent
A scalar function $h(\cdot,\cdot)$ of two vectors $\ba\in{\mathbb
  R}^m$ and $\bb\in{\mathbb R}^m$ is defined as isotropic with respect
to any orthogonal transformation $\bQ^{(m\times m})$ if $h(\ba, \bb) =
h(\bQ\ba, \bQ\bb)$. Here $\bQ^T\bQ=\bI$, the identity matrix and ${\textrm{det}}\bQ=\pm 1$.
Under any such orthogonal transformation $\bQ$, only the
magnitudes of the vectors $\ba$ and $\bb$, and the angle between them
remain invariant, where the angle between $\ba$ and $\bb$ is
$\displaystyle{\frac{\ba\cdot\bb}{\sqrt{\ba\cdot\ba}\sqrt{\bb\cdot\bb}}}$. Therefore,
it follows that \\
$$h(\cdot,\cdot)\quad\mbox{is isotropic}\quad \iff$$
$$h(\ba, \bb)=h(\bQ\ba, \bQ\bb)=h(\ba\cdot\ba, \bb\cdot\bb, \ba\cdot\bb)$$.
\end{remark}

We also recall that in this application, $\bX\cdot\bV=0$ by construction.

This leads us to identify any $pdf$ of the state space variable $\bW=(\bX^T,\bV^T)^T$ as isotropic if the $pdf$ is expressed as a function of energy $E(\bX,\bV)$ alone. This follows from Equation~\ref{eqn:fmanyforms} since 
$f(\bX,\bV) = f(E)\Longrightarrow$
\begin{eqnarray}
f(\bX,\bV) &=& f(\Phi(\rho(R)),\bV\cdot\bV)\\ \nonumber
           &=& f(\bX\cdot\bX,\bV\cdot\bV,\bX\cdot\bV) \\ \nonumber
           && \quad\mbox(since \bX\cdot\bV=0)
\end{eqnarray}
which is compatible with the form of isotropic functions as per Remark~\ref{remark:isotropic}. Thus, if the $pdf$ of the state space variable is dependent on only 1 constant of motion--which by the literature in galactic dynamics has to be energy $E(\bX,\bV)$--then $f(\bX,\bV)$ is an isotropic function of $\bX$ and $\bV$.

However, there is no prior reason to model a real galaxy as having an
isotropic probability distribution of its state space. Instead, we
attempt to 
\begin{itemize}
\item use as general a model for the state space distribution of
the system as possible, 
\item while ensuring that the degrees of freedom in
the model are kept to a minimum to ease computational ease. 
\end{itemize}
This leads us to include another time-invariant function $L(\bX,\bV)$
in the definition of the $pdf$ of the state space variable in addition
to $E(\bX,\bV)$, such that the dependence on $\bX$ and $\bV$ in
$L(\cdot,\cdot)$ is not of the form that renders $f(E,L)$ compatible
with the definition of isotropic function, as per
Remark~\ref{remark:isotropic}, unlike $f(E)$.

This is so because
\begin{equation}
L(\bX,\bV) := \parallel \bX \btimes \bV \parallel
\label{eqn:den_l}
\end{equation}
where $\btimes$ represents the ``cross-product'' of the two 3-dimensional vectors $\bX$ and $\bV$, i.e.\\
\begin{equation}
(\bX\times\bV)^T:=
{
\left(\left| \begin{array}{cc}
X_2 & X_3 \\
V_2 & V_3 \end{array} \right| - \left| \begin{array}{cc}
X_1 & X_3 \\
V_1 & V_3 \end{array} \right|  + \left| \begin{array}{cc}
X_1 & X_2 \\
V_1 & V_2 \end{array} \right|\right)}
\end{equation}
so that \\
\begin{equation}
L(\bX,\bV) = 
\parallel(X_2 V_3-X_3 V_2, X_3 V_1-X_1 V_3, X_1 V_2-X_2 V_1)^T\parallel
\label{eqn:ldefn}
\end{equation}
Then, we set
$f(E,L) \equiv f(\bX\cdot\bX, \bV\cdot\bV, \bX \btimes \bV)$ which is not compatible with the form of an isotropic function of the 2 vectors $\bX$ and $\bV$. In other words, if the support of the $pdf$ of $\bX$ and $\bV$ includes $E(\bX,\bV)$ and $L(\bX,\bV)$, then the state space distribution is no longer restricted to be isotropic. 

Such a general state space is indeed what we aimed to achieve with our
model. At the same time, adhering to no more than 1 constant of motion
in addition to energy $E(\bX,\bV)$ helps to keep the dimensionality of
the domain of the $pdf$ of the state space function to the minimum that
it can be, given our demand that no stringent model-driven constraint
be placed on the state space geometry. Thus, we use $n$=2 in our model.

So now we are ready to express the unknown gravitational mass density
function as embedded within the $pdf$ of $\bX$ and $\bV$ as:\\
$f(\bX,\bV)=$
\begin{equation}
f(E(\Phi(\rho(R))+\bV\cdot\bV/2),L(\bX,\bV))\equiv 
\small{f(\Phi(\rho(\sqrt{X_1^2+X_2^2+X_3^2})), V_1^2+V_2^2+V_3^2, \parallel\bX\btimes\bV\parallel)}
\label{eqn:inter}
\end{equation}
using Equation~\ref{eqn:den_l}. To cast this in the form of Equation~\ref{eqn:f_fin}, we realise that the unknown gravitational mass density function will need to be discretised; we would first discretise the range of values of $R$ over which the
gravitational mass density function $\rho(R)$ is sought. Let $R=r$ such that
$r\in[r_{min}, r_{max}]$
and let the width of each $R$-bin be
$\delta_r$. Then $\rho(r)$ is discretised as the unknown model parameter vector 
\begin{equation}
\brho := (\rho_1,\rho_2,\ldots,\rho_{N_x})^T,
\end{equation}  
where
\begin{equation}
\rho_b:=\rho(r) \quad\mbox{for}\quad r\in[(b-1)\delta_r, b\delta_r]\quad b=1,2,\ldots,N_x
\label{eqn:rho_discrete}
\end{equation}
where $N_x:=\displaystyle{{\textrm{int}}\left(\frac{r_{max}-r_{min}}{\delta_r}\right)+}$. 

Then following on from Equation~\ref{eqn:inter} we write
\begin{equation}
f(\bX,\bV) = f(\brho, V_1^2+V_2^2+V_3^2, \parallel\bX\btimes\bV\parallel)
\end{equation}
This is in line with Equation~\ref{eqn:f_fin} if we identify the function of the unknown model parameter vector $\bxi(\rho)$ in the RHS of Equation~\ref{eqn:f_fin} with the unknown gravitational mass density vector $\brho$. Then the $pdf$ of the state space variables $\bX$ and $\bV$ depends of $\brho$ and $\bX$ and $\bV$. Then the equivalent of Equation~\ref{eqn:nu_second} is
\begin{equation}
\nu(x_1^{(k)},x_2^{(k)},v_3^{(k)},\brho)={\displaystyle{\int\limits_{X_3}\int\limits_{V_{2}}\int\limits_{V_1} 
f\left(\brho, (v_1^2)^{(k)}+(v_2^2)^{(k)}+v_3^2, \parallel(x_1^{(k)},x_2^{(k)},x_3)^T\btimes(v_1,v_2,v_3^{(k)})^T\parallel\right) dx_{3} dv_{2} dv_{1}}},
\label{eqn:nu_galaxy}
\end{equation}

$k=1,2,\ldots,N_{data}$.
Then plugging this in the RHS of Equation~\ref{eqn:likeli_prelim}, the likelihood is 
\begin{equation}
\Pr(\{\bu^{(k)}\}_{k=1}^{N_{data}}\vert\brho) = \displaystyle{\prod_{k=1}^{N_{data}} \nu(\bu^{(k)},\brho)}
\label{eqn:likeli_galaxy}
\end{equation}
Then to compute the likelihood and thereafter the posterior probability of $\brho$ given data ${\bf D}$, we will need to compute the integral in Equation~\ref{eqn:nu_galaxy}. According to the general methodology discussed above in Section~\ref{sec:generic}, this is performed by discretising the domain of the $pdf$ of the state space variable, i.e. of $f(E,L)$. In order to achieve this discretisation we will need to invoke the functional relationship between $E(\bX,\bV)$ and $L(\bX,\bV)$. Next we discuss this.


\subsection{Relationship between $E(\bX,\bV)$ and $L(\bX,\bV)$}
\label{sec:relation}
\noindent
We recall the physical interpretation of $L(\bX,\bV)$ as the norm of the ``angular momentum'' vector, i.e. $\displaystyle{\frac{L^2(\bX,\bV)}{R^2}}$ is the square of the speed $V_c(\bX,\bV)$ of circular motion of a particle with location $\bX$ and velocity $\bV$; here, ``circular motion'' is motion orthogonal to the location vector $\bX$, distinguished from non-circular motion that is parallel to $\bX$ and the speed of which is $V_{nc}(\bX,\bV)$. Then as these two components of motion are mutually orthogonal, square of the particle's speed is 
\begin{equation}
V^2\equiv \bV\cdot\bV = V_c^2 + V_{nc}^2 = \displaystyle{\frac{L^2(\bX,\bV)}{R^2} + V_{nc}^2},
\end{equation}
where $V_{nc}$ is the magnitude of the component of $\bV$ that is parallel to $\bX$, i.e. 
\begin{equation}
V_{nc} = \displaystyle{\frac{\bV \cdot \bX}{\parallel \bX \parallel}}
\end{equation}
But we recall that energy $E(\bX,\bV)=\Phi(\rho(R)) + \bV\cdot\bV/2$. 

This implies that
\hspace*{.6cm}
\begin{eqnarray}
E(\bX,\bV)=&&\Phi(\rho(R)) + \displaystyle{\frac{V^2}{2}}\nonumber \\ 
=&&\Phi(\rho(R)) + \displaystyle{\frac{L^2(\bX,\bV)}{2R^2} + \frac{V_{nc}^2}{2}} \nonumber \\
=&&\Phi(\rho(R)) + \displaystyle{\frac{V_{nc}^2}{2}} +
{\displaystyle{\small{\frac{\parallel(X_2 V_3-X_3 V_2, X_3 V_1-X_1 V_3, X_1 V_2-X_2 V_1)^T\parallel^2}{2R^2}}}}\nonumber \\
&& 
\label{eqn:ljhamela}
\end{eqnarray}
where in the last equation, we invoked the definition of $L(\bX,\bV)$ sing Equation~\ref{eqn:ldefn}.

At this stage, to simplify things, we consciously choose to work in the coordinate system in which the vector $\bX$ is rotated to vector $\bS=(S_1,S_2,S_3)^T$
by a rotation through angle $\theta:=\displaystyle{\cos^{-1}\frac{X_2}{\sqrt{X_1^2+X_2^2}}}$, i.e.
\begin{equation}
 \left(\begin{array}{c}
S_1 \\
S_2 \\
S_3 \end{array} \right)
=\left(\begin{array}{ccc}
\cos\theta & -\sin\theta & 0 \\
\sin\theta & \cos\theta & 0 \\
0 & 0 & 1 \end{array} \right)
 \left(\begin{array}{c}
X_1 \\
X_2 \\
X_3 \end{array} \right)
\end{equation}
Then by definition, $S_1$=0, i.e. the projection of the
$(S_1,S_2,S_3)^T$ vector on the $S_3$=0 plane lies entirely along the
$S_2$-axis.

This rotation does not affect the previous
discussion since 
\begin{itemize}
\item the previous discussion invokes the location variable
either via $R=\sqrt{X_1^2+X_2^2+X_3^2}=\sqrt{S_1^2+S_2^2+S_3^2}$, 
\item or via $\sqrt{x_1^2 + x_2^2}=\sqrt{s_1^2+s_2^2}$ as within the data
structure: ${\bf D}=\left\{\sqrt{(x_1^{(k)})^2+ (x_2^{(k)})^2},
v_3^{(k)}\right\}_{k=1}^{N_{data}}=$ \\
$\left\{\sqrt{(s_1^{(k)})^2+ (s_2^{(k)})^2},v_3^{(k)}\right\}_{k=1}^{N_{data}}\equiv\left\{s_1^{(k)}, s_2^{(k)}, v_3^{(k)}\right\}_{k=1}^{N_{data}}$.
\end{itemize} 
Having undertaken the rotation, we refer to $E(\bX,\bV)$ and
$L(\bX,\bV)$ as $E(\bS,\bV)$ and $L(\bS,\bV)$ respectively.

This rotation renders the cross-product in the definition of $L(\cdot,\cdot)$ simpler; under this choice of the
coordinate system, as $S_1=0$ 
\begin{eqnarray}
[L(\bS,\bV)]^2 &=& {\parallel \bS\btimes \bV\parallel^2}\nonumber\\
&=&{\parallel(S_2 V_3 - S_3 V_2, S_3 V_1, -S_2 V_1)^T\parallel^2}\nonumber\\ 
&=&{\parallel(R V_3 \sin\gamma - R V_2\cos\gamma, R V_1\cos\gamma, -R V_1\sin\gamma)^T\parallel^2} \nonumber\\
&=&{R^2\left[V_1^2 + (V_2\cos\gamma - V_3\sin\gamma)^2\right]}
\label{eqn:lcoord}
\end{eqnarray}
where
\begin{eqnarray}
\displaystyle{\frac{S_3}{R}=\frac{S_3}{\sqrt{S_1^2+S_2^2+S_3^2}}}:=\cos\gamma
\end{eqnarray} 
so that $\displaystyle{\frac{S_2}{\sqrt{S_1^2+S_2^2+S_3^2}}=\sin\gamma=\frac{S_2}{R}}$, 
so that in this rotated coordinate system, from Equation~\ref{eqn:ljhamela}
\begin{eqnarray}
E(\bS,\bV) &=& \Phi(\rho(R)) \nonumber \\ 
{} &&+ \displaystyle{\frac{\left[V_1^2 + (V_2\cos\gamma - V_3\sin\gamma)^2\right]}{2}} \nonumber \\
{} &&+ \displaystyle{\frac{V_{nc}^2}{2}}.
\label{eqn:elrevamp}
\end{eqnarray}
Also, the component of $\bV$ along the location vector $\bS$ is $V_{nc}=\bV\cdot\bS/R=(V_2 S_2+V_3 S_3)/R$.

From Equation~\ref{eqn:ljhamela} it is evident that for a given value $\epsilon$ of $E(\bS,\bV)$, the highest value $\ell_{max}(\epsilon)$ of $L(\bS,\bV)$ is attained if $V_{nc}=0$ (all motion is circular motion). This is realised only when the radius $R_c$ of the circular path of the particle takes a value $r_c$ such that
\begin{equation}
\displaystyle{{\ell_{max}(\epsilon)^2}} = \displaystyle{2r_c^2
  \left[\epsilon-\Phi(r_c)\right]}
\label{eqn:ellmax}
\end{equation}
The way to compute $r_c$ given $\epsilon$ is defined in the literature
\ctp{BT} as the positive definite solution for $r$ in the equation
\begin{equation}
\displaystyle{2r^2 \left[\epsilon-\Phi(r)\right]} = \displaystyle{-r^3\frac{d\Phi(r)}{dr}}
\label{eqn:ellsoln}
\end{equation}

We are now ready to discretise the domain of the $pdf$ of the state space variable, i.e. of $f(E,L)$ in line with the general methodology discussed above in Section~\ref{sec:generic} with the aim of computing the integral in Equation~\ref{eqn:nu_galaxy}.

\subsection{Discretisation of $f(E,L)$}
\label{sec:discretef}
\noindent
We discretise the domain of $fE,L)$ where this 2-dimensional domain is defined by the range of values $E=\epsilon\in[\epsilon_{min},\epsilon_{max}]$ and $L=\ell\in[\ell_{min},\ell_{max}]$, by placing a uniform 2-dimensional rectangular grid over $[\epsilon_{min},\epsilon_{max}]\times[\ell_{min},\ell_{max}]$ such that the range $[\epsilon_{min},\epsilon_{max}]$ is broken into $E$-bins each $\delta_\epsilon$ wide and the range $[\ell_{min},\ell_{max}]$ is broken into $L$-bins each $\delta_\ell$ wide. Then each 2-dimensional $E-L$-grid cell has size $\delta_\epsilon\times\delta_\ell$. Then,
\begin{eqnarray}
f_{c,d} &:=& f(\epsilon,\ell) \quad\mbox{{for}}\nonumber \\ 
\epsilon&\in&[\epsilon_{min}+(c-1)\delta_\epsilon, \epsilon_{min}+c\delta_\epsilon],\nonumber \\
\ell&\in&[\ell_{min}+(d-1)\delta_\ell, \ell_{min}+d\delta_\ell],\nonumber \\
c&=&1,2,\ldots,N_\epsilon,\nonumber \\
d&=&1,2,\ldots,N_\ell,
\label{eqn:prelim_cd}
\end{eqnarray}
where the number of $E$-bins is $N_\epsilon:=\displaystyle{{\textrm{int}}\left(\frac{\epsilon_{max}-\epsilon_{min}}{\delta_\epsilon}\right)+1}$ and the number of $L$-bins is  $N_\ell:=\displaystyle{{\textrm{int}}\left(\frac{\ell_{max}-\ell_{min}}{\delta_\ell}\right)+1}$. We then define the $N_\epsilon\times N_\ell$-dimensional matrix 
\begin{equation}
{\bf F} := [f_{c,d}]_{N_\epsilon\times N_\ell}.
\label{eqn:matis_pdf}
\end{equation}
In our model this is the discretised version of the $pdf$ $f(E,L)$ of
the state space variable $\bW=(\bS^T,\bV^T)^T$. 

In this application, a particle with a positive value of energy is so
energetic that it escapes from the galaxy. We are however
concerned with particles that live inside the galaxy, i.e. are bound to the galaxy and therefore, the maximum
energy that a galactic particle can attain is 0, i.e. $\epsilon_{max}=0$. Given the definition of energy $E(\bS,\bV)=\Phi(R) + \bV\cdot\bV/2$ we realise that the value of $E(\bS,\bV)$ is minimum,
i.e. as negative as it can be, if $\bV\cdot\bV$=0, (i.e. velocity is
zero) and $\Phi(R)$ is minimum, which occurs at $R=0$. In
other words, the minimum value of $E$ is $\Phi(0)$ which is
negative. In our work we normalise the value $\epsilon$ of $E$ by
$-\Phi(0)$, so that $\epsilon\in[-1,0]$. In other words, the
aforementioned $\epsilon_{min}=-1$ and $\epsilon_{max}=0$.

We normalise the value $\ell$ of $L(\bS,\bV)$ with the maximal value
$\ell_{max}(\epsilon)$ that $\ell$ can attain for a given value
$\epsilon$ of $E$ (Equation~\ref{eqn:ellmax}). The maximum value that can be attained by $L$ is for $\epsilon=0$; having computed $r_c$ from
Equation~\ref{eqn:ellsoln}, $\ell_{max}(0)$ is computed. Then, as
normalised by $\ell_{max}(0)$, the maximal value of $L$ is 1. Also the
lowest value of $L$ is 0, i.e. $\ell_{min}$=0.  In light of this, we
rewrite Equation~\ref{eqn:prelim_cd} as
\begin{eqnarray}
f_{c,d} &:=& f(\epsilon,\ell) \quad\mbox{{for}}\nonumber \\ 
\epsilon&\in&[-1+(c-1)\delta_\epsilon, -1+c\delta_\epsilon],\nonumber \\
\ell&\in&[(d-1)\delta_\ell, d\delta_\ell],\nonumber \\
c&=& 1,2,\ldots,N_\epsilon,\nonumber \\
d&=& 1,2,\ldots,N_\ell.
\label{eqn:prelim_cd2}
\end{eqnarray}
The $E$-binning and $L$-binning are kept uniform in the application we
discuss below, i.e. $\delta_\epsilon$ and $\delta_\ell$ are constants.

\subsubsection{Data-driven binning}
\noindent
There are $N_\ell$ $L$-bins and $N_\epsilon$ $E$-bins.  Above we saw
that as the range covered by normalised values of $E$ is $[-1,0]$, the
relationship between $N_\epsilon$ and $E$-bin width $\delta_\epsilon$
is $\delta_\epsilon=1/N_\epsilon$. We make inference on $N_\ell$
within our inference scheme while the Physics of the situation drives
us to a value of $N_\epsilon$. It could have been possible to also
learn $N_\epsilon$ from the data within our inference scheme but that
would have been tantamount to wastage of information that is available
from the domain of application.

We attempt to learn $N_\ell$ from the data within our inference
scheme; for a given $N_\ell$, $N_\epsilon$ is fixed by the data at
hand. To understand this, we recall the aforementioned relation
$\epsilon=\Phi(r) + v_{nc}^2/2 + \ell^2/2r^2$. Let in the available
data set,
\begin{enumerate}
\item[--]the minimum value of $\sqrt{S_1^2+S_2^2}$ be $r_{min}$,
\item[--]the maximum value of $\sqrt{S_1^2+S_2^2}$ be $r_{max}$ so that the value of $\Phi(\cdot)$ is no less than $\Phi(r_{max})$, 
\item[--]the maximum value of $V_3$ be $v_3^{(max)}$ so that the unnormalised value of $E$ is no less than 
\begin{equation}
\epsilon_{max}:=\displaystyle{\Phi(r_{max}) + \frac{[v_3^{(max)}]^2}{2} + \frac{[N_\ell \ell_{max}(0)]^2}{2 r_{min}^2}}
\end{equation}
\item[--]and the unnormalised $\epsilon$ is no more than $\Phi(0)$.
\end{enumerate}
Thus, it is clear that the $E$-binning should cover the interval
beginning at a normalised value of -1 and should at least extend to
$\epsilon_{max}/[-\Phi(0)]$. 

Then we set $E$-bin width $\delta_\epsilon=1/N_\epsilon$
and learn number of $L$-bins, $N_\ell$, from the data within
our inference scheme. Then at any iteration, for the current value of
$N_\ell$ and the current $\brho$ (which leads to the current value of
$\Phi(r)$ according to Equation~\ref{eqn:poisson}), placing
$\epsilon_{max}/[-\Phi(0)]$ at the centre of the $N_\epsilon$-th
$E$-bin gives us
\begin{equation}
\displaystyle{\frac{\epsilon_{max}}{-\Phi(0)}} = -1 + (N_\epsilon - 0.5)\delta_\epsilon
\end{equation}
i.e. $N_\epsilon = {\textrm{int}}\left(\Phi(0)/[2\epsilon_{max}]\right)$.

Experiments suggest that for typical galactic data sets, $N_\ell$
between 5 and 10 implies convergence in the learnt vectorised form of
the gravitational mass density $\brho$. This leads us to choose a
discrete uniform prior over the set $\{5,6,\ldots,10\}$, for $N_\ell$:
\begin{equation}
\pi_0(N_\ell)=\displaystyle{\frac{1}{5}}.
\label{eqn:priornell}
\end{equation}

Again, the minimum and maximum values of $\sqrt{S_1^2+S_2^2}$ in the
data fix $r_{min}$ and $r_{max}$ respectively, so that $r_{max} =
r_{min} + \delta_r(N_x -1)$.  The radial bin width $\delta_r$ is
entirely dictated by the data distribution such that there is at least
1 data vector in each radial bin. Thus, $N_x$ and $\delta_r$ are not
parameters to be learnt within the inference scheme but are directly
determined by the data.

\subsection{Likelihood}
\label{sec:appl_likeli}
\noindent
Following Equation~\ref{eqn:likeli_fin}, we express the likelihood in this application in terms of the $pdf$ of $\bS$ and $\bV$, marginalised over all those variables that we do not have any observed information on. Then for the data vector $(s_1^{(k)}, s_2^{(k)}, v_3^{(k)})^T$, the marginal $pdf$ is\\ 
{{
\begin{equation}
\nu(s_1^{(k)}, s_2^{(k)}, v_3^{(k)}) =
\displaystyle{\small{\int\limits_{S_3}\int\limits_{V_1}\int\limits_{V_2} f\left(g_1(r^{(k)},v_3^{(k)},v_1,v_2),g_2(r^{(k)},\gamma^{(k)},v_3^{(k)},v_1,v_2)\right) ds_3 dv_1 dv_2}},\nonumber 
\end{equation}}}
where\\ 
$g_1(r^{(k)},v_3^{(k)},v_1,v_2):=\small{\Phi(r^{(k)})+ \frac{[v_1^2+v_2^2+(v_3^{(k)})^2]}{2}}$\\
$g_2(r^{(k)},\gamma^{(k)},v_3^{(k)},v_1,v_2):=$\\
$\small{[(r^{(k)})^2(v_1^2 + (v_2\cos\gamma^{(k)}-v_3^{(k)}\sin\gamma^{(k)}))^2]}$,\\
with $[L(\bS,\bV)]^2$ recalled from Equation~\ref{eqn:lcoord}, and we have used
\begin{equation}
r^{(k)}:=
\sqrt{(s_1^{(k)})^2+(s_2^{(k)})^2+s_3^2}
\label{eqn:rk}
\end{equation}
and
$\cos\gamma^{(k)}:=\displaystyle{\frac{s_3}{r^{(k)}}}$.

Then given that the range of values of $E$ and $L$ is discretised, we write
\begin{equation}
\nu(s_1^{(k)}, s_2^{(k)}, v_3^{(k)}) = 
\displaystyle{\sum_{c=1}^{N_\epsilon}\sum_{d=1}^{N_\ell}\left[f_{c,d}\int_{\{s_3^{(c,d)}\}\vert{\brho},\{v_1^{(c,d)}\}\vert{\brho},\{v_2^{(c,d)}\}\vert{\brho}} ds_3 dv_1 dv_2\right]},
\label{eqn:likeli_integ}
\end{equation}
where $\{s_3^{(c,d)}\}\vert\brho$ refer to the values taken by $S_3$ for a given $\brho$, inside the
$cd$-th $E-L$-grid-cell.
Similarly, $\{v_1^{(c,d)}\}\vert\brho$ and $\{v_2^{(c,d)}\}\vert\brho$
refer to the values of $V_1$ and $V_2$ inside the $cd$-th
$E-L$-grid-cell respectively, given $\brho$. 

Indexing the values of any of the unobserved variables in this
grid-cell as conditional on $\brho$, is explained as
follows. $\{s_3^{(c,d)}\}$, $\{v_1^{(c,d)}\}$ and $\{v_2^{(c,d)}\}$
are determined by the mapping between the space of $E$ and $L$ and the
space of the unobservables, namely $S_3,V_1,V_2$. This
mapping involves the definition of $E$ and $L$ in terms of the state
space coordinates $(\bS^T,\bV^T)^T$, which in turn depends upon the
function $\rho(r)$ or its discretised version, $\brho$. Hence the
values taken by any of the 3 unobservables in the $cd$-th
$E-L$-grid-cell depend on $\brho$.  Here
$c=1,2,\ldots,N_\epsilon$ and $d=1,2,\ldots,N_\ell$.

We realise that the integral on the RHS of
Equation~\ref{eqn:likeli_integ} represents the volume occupied by
the $E-L$-grid-cell inside the space of the unobserved variables. The
computation of this volume is now discussed.

\subsection{Volume of any $E-L$-grid-cell in terms of the unobservables}
\label{sec:mapping}
\noindent
We begin by considering the volume of any $E-L$-grid-cell in the space of the 2 observables, $V_1$ and $V_2$, at a given value of $S_3$. Thereafter, we will consider the values of the 3rd unobservable, $S_3$, in this grid-cell.

The definition $E(\bs^{(k)},\bv^{(k)})=\Phi(r^{(k)}) +
\bv^{(k)}\cdot\bv^{(k)}/2$ (Equation~ref{eqn:ljhamela}) implies that
for the $k$-th data vector $(s_1^{(k)}, s_2^{(k)}, v_3^{(k)})^T$, all
particles with $S_3=s_3$ and energy
$E(\bs^{(k)},\bv^{(k)})=\epsilon_c$ will obey the equation
\begin{equation}
v_1^2 + v_2^2 = 2\left[\epsilon_c - \Phi(r^{(k)})\right] - (v_3^{(k)})^2,
\label{eqn:engcirc}
\end{equation}
i.e. for $S_3=s_3$, all particles lying in the $c$-th E-bin will lie in the space of $V_1$ and $V_2$, within a circular annulus that is centred at (0,0) and has radii lying in the interval $[\varepsilon_{c+1}, \varepsilon_c]$ where
\begin{eqnarray}
\varepsilon_{c+1} &:=& \displaystyle{\sqrt{\left\{2\epsilon_{c+1} - 
2\Phi(r^{(k)}) - (v_3^{(k)})^2\right\}}}\nonumber \\ 
\varepsilon_{c} &:=& \displaystyle{\sqrt{\left\{2\epsilon_{c} - 
2\Phi(r^{(k)}) - (v_3^{(k)})^2\right\}}}. 
\end{eqnarray}


For $S_3=s_3$, the definition $L(\bs^{(k)},\bv^{(k)})=\parallel
\bs^{(k)}\btimes \bv^{(k)}\parallel$ provides a representation for all
particles in the $d$-th $L$-bin with given observed values of $S_1$,
$S_2$ and $V_3$. 


It then follows from $[L(\bs^{(k)},\bv^{(k)})]^2=
(r^{(k)})^2\left[(v_1^2 + \{v_2\cos\gamma -v^{(k)}_3\sin\gamma\}^2\right]$, 
(Equation~\ref{eqn:lcoord})
that for the $k$-th data vector, all particles with $S_3=s_3$,
and in the $d$-th $L$-bin ($L(\bs^{(k)},\bv^{(k)})=\ell_d$) will
obey the equation
\begin{equation}
\displaystyle{\frac{\ell_d^2}{(s_1^{(k)})^2+(s_2^{(k)})^2+s_3^2}} = \left[v_1^2 + \cos^2\gamma^{(k)}(v_2 - v^{(k)}_3\tan\gamma^{(k)})^2\right].
\label{eqn:ellell}
\end{equation}
where we have recalled $r^{(k)}$ from Equation~\ref{eqn:rk}.
This implies that for $S_3=s_3$, all particles lying in the $d$-th L-bin,
will lie in the space of $V_1$ and $V_2$,  along an ellipse centred at 
$(0, v_3^{(k)}\tan\gamma^{(k)})$ with semi-minor axis lying in the interval of
$[\lambda_{d+1},\lambda_d]$ and semi-major axis lying in the interval
$\displaystyle{\left[\frac{\lambda_{d+1}}{\cos\gamma^{(k)}},\frac{\lambda_d}{\cos\gamma^{(k)}}\right]}$. Here,
\begin{eqnarray}
\lambda_{d+1} &:=& \displaystyle{\frac{\ell_{d+1}}{\sqrt{(s_1^{(k)})^2+(s_2^{(k)})^2+s_3^2}}} \nonumber \\
\lambda_{d} &:=& \displaystyle{\frac{\ell_{d}}{\sqrt{(s_1^{(k)})^2+(s_2^{(k)})^2+s_3^2}}}
\end{eqnarray}

Collating the implications of Equation~\ref{eqn:engcirc} and
Equation~\ref{eqn:ellell}, we get that at a given value of $S_3$,
particles with observed data $(s_1^{(k)},s_2^{(k)},v_3^{(k)})^T$, (with
energies) in the $c$-th $E$-bin and (momenta) in the $d$-th $L$-bin will lie
in the space of $V_1$ and $V_2$, within an area bound by the overlap
of
\begin{enumerate}
\item[--] the circular annular region centred at $V_1=0, V_2=0$, extending in radii between $\varepsilon_{c+1}$ and $\varepsilon_c$.
\item[--] the elliptical annular region centred at $V_1=0, V_2=v_3^{(k)}\tan\gamma$, extending in semi-minor between $\lambda_{d+1}$ and $\lambda_d$ and semi-major axis in $[\lambda_{d}/\cos\gamma,\lambda_{d+1}/cos\gamma]$, where $\cos\gamma=\displaystyle{\frac{s_3}{\sqrt{(s_1^{(k)})^2+(s_2^{(k)})^2+s_3^2}}}$.
\end{enumerate}
The area of these overlapping annular regions represents the volume of
the $cd$-th $E-L$-grid-cell in the space of $V_1$ and $V_2$, at the
value $s_3$ of $S_3$. Thus, the first step towards writing the volume 
of the $cd$-th $E-L$-grid-cell in terms of the unobservables, is to compute the area of these overlapping annular regions in the space of $V_1$ and $V_2$. Such an area of overlap is a function of $s_3$. At the next step, we integrate such an area over all allowed $s_3$, to recover the volume of the $cd$-th $E-L$-grid-cell in the space of $V_1$, $V_2$ and $S_3$, i.e. the integral on the RHS of Equation~\ref{eqn:likeli_integ}.

There can be multiple ways these annular regions overlap; three
examples of these distinct overlapping geometries are displayed in
Figure~\ref{fig:overlap}. In each such geometry, it is possible to
compute the area of this region of overlap since we know the equations
of the curves that bound the area. However, the number of possible
geometries of overlap is in excess of 20 and identifying the
particular geometry to then compute the area of overlap in each such
case, is tedious to code. In place of this, we allow for a numerical
computation of the area of overlap; this method works irrespective of
the particulars of the geometry of overlap. We identify the maximum
and minimum values of $V_2$ allowed at a given value of $V_1$, having
known the equations to the bounding curves, and compute the area of
overlap in the plane of $V_1$ and $V_2$ using numerical integration.

This area of overlap in the plane defined by $V_1$ and $V_2$ is a
function of $S_3$ since the equations of the bounding curves are
expressed in terms of $s_3$. The area of overlap is then integrated
over all values that $S_3$ is permitted to take inside the $cd$-th
$E-L$-grid-cell. For any $E-L$-grid-cell, the lowest value $S_3$ can
take is zero. For $\epsilon\in[\epsilon_{c+1},\epsilon_{c}]$, and $\ell\in[\ell_d,\ell_{d+1}]$, the
maximum value of $S_3$ is realised (by recalling Equation~\ref{eqn:elrevamp}) as the solution to
the equation
\begin{equation}
\displaystyle{2(\epsilon_{c} - \Phi(r))} = \displaystyle{\frac{\ell_d^2}{r^2} + v_{nc}^2}
\label{eqn:zeqn}
\end{equation}

\begin{figure}[ht!]
     \begin{center}
{
   {
    \includegraphics[width=10cm]{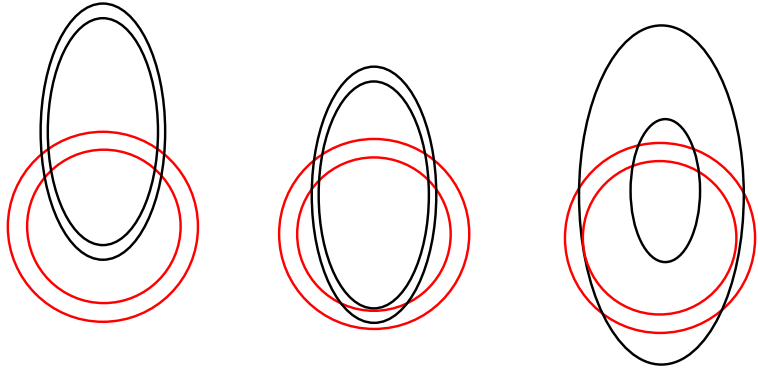} 
   }
}
\end{center}
\caption{Figure showing 3 of the many ways of overlap between the contours drawn in the space of $V_1$ and $V_2$, at neighbouring values of $E$ (the circular contours in red) and at neighbouring values of $L$ (the elliptical contours in black).}
\label{fig:overlap}
\end{figure}

where $v_{nc}$ is the projection of $\bv$ along the $\bs$ vector (discussed in Section~\ref{sec:relation}). Thus, $v_{nc}$ is given by the inner product of $\bv$ and the unit vector parallel to $\bs$:
\begin{equation}
v_{nc} = \displaystyle{\frac{\bv\cdot\bs}{\parallel\bs\parallel}}, 
\label{eqn:vr}
\end{equation}
where $\parallel\bs\parallel\equiv r$. Under our choice of coordinate system,
Equation~\ref{eqn:vr} gives
\begin{eqnarray}
v_{nc} &=& \displaystyle{\frac{v_2 s_2}{r} + \frac{v_3 s_3}{r}}\nonumber \\
  &=& \displaystyle{{v_2 \sin\gamma} + v_3\cos\gamma}\quad\mbox{where}\nonumber \\
\cos\gamma &:=& \displaystyle{\frac{s_3}{r}}
\end{eqnarray}
Using this in Equation~\ref{eqn:zeqn} we get
\begin{equation}
\displaystyle{2r^2(\epsilon_c - \Phi(r))} = \displaystyle{{\ell_d^2} + v_2^2 s_2^2 + v_3^2 s_3^2 + 2v_3 v_2 s_2 s_3}.
\label{eqn:zeqn2}
\end{equation}
This implies that given the observations represented by the $k$-th data vector $(s_1^{(k)}, s_2^{(k)}, v_3^{(k)})$, 
\begin{eqnarray}
\displaystyle{2\left[(s_1^{(k)})^2 + (s_2^{(k)})^2 + s_3^2\right]\left[\epsilon_c - \Phi(r)\right]} =&&\nonumber \\ 
\displaystyle{{\ell_d^2} + v_2^2 (s_2^{(k)})^2 + (v_3^{(k)})^2 s_3^2 + 2v_3^{(k)} v_2 s_2^{(k)} s_3}.&&
\label{eqn:zeqn3}
\end{eqnarray}
The highest positive root for $s_3$ from Equation~\ref{eqn:zeqn3} as the highest value that $S_3$ can attain in the $cd$-th $E-L$-grid-cell. Thus, for the $cd$-th cell, the limits on the integration over $s_3$ are 0 and the solution to Equation~\ref{eqn:zeqn3}. 

So now we have the value of the integral over $v_1$ and $v_2$ and
hereafter over $s_3$, for the $cd$-th $E-L$-grid-cell. This triple
integral gives the volume of the $cd$-th $E-L$-grid-cell in the space
of the unobservables, i.e. of $V_1, V_2, S_3$. This volume is
multiplied by the value $f_{c,d}$ of the discretised $pdf$ of the state space variable
in this $E-L$ cell and the resulting product is summed over all $c$
and $d$, to give us the marginalised $pdf$ $\nu(s_1^{(k)},
s_2^{(k)}, v_3^{(k)})$ (see Equation~\ref{eqn:likeli_integ}). Once the
marginalised $pdf$ is known for a given $k$, the product over all $k$s
contributes towards the likelihood.


\subsection{Normalisation of the marginal $pdf$ of the state space vector}
\noindent
As we see from Equation~\ref{eqn:likeli_integ}, the marginal $pdf$ of
$\bS$ and $\bV$ is dependent on $\brho$, so this normalisation will
not cancel within the implementation of Metropolis-Hastings to perform
posterior sampling. In other words, to ensure that the value of
$\nu(\cdot,\cdot,\cdot)$ - and therefore the likelihood - is not
artificially enhanced by choosing a high $\brho$, we normalise
$\nu(s_1^{(k)}, s_2^{(k)}, v_3^{(k)})$ for each $k$, by the $pdf$
integrated over all possible values of $S_1$, $S_2$ and $V_3$, i.e. by
\begin{eqnarray}
\displaystyle{\int_{S_1}\int_{S_2}\int_{V_3} \nu(s_1,s_2,v_3) ds_1 ds_2 dv_3}
\end{eqnarray}
where the possible values of $V_3$ are in the interval
$[-\sqrt{-2\Phi(s_1^2+s_2^2)}, \sqrt{-2\Phi(s_1^2+s_2^2)}]$, of $S_2$
in the interval $[\sqrt{s_1^2 - r_{min}^2}, \sqrt{r_{max}^2-s_1^2}]$
and of $S_1$ in $[r_{min},r_{max}]$. Hereafter, by $\nu(\cdot,\cdot,\cdot)$
we will imply the normalised marginal $pdf$.

\subsection{Incorporating measurement uncertainties}
\noindent
Following Equation~\ref{eqn:likeli_fin2} the likelihood is defined as
the product over all data, of the convolution of the error
distribution at the $k$-th datum and value of the marginalised $pdf$
for this $k$ (assuming the data to be conditionally $iid$). In this
application the measurement of the location of the galactic particle
projected onto the image plane of the galaxy, i.e. $(S_1,S_2)$, is
constrained well enough to ignore measurement uncertainties
in. However, the measurement errors in the line-of-sight component of
the particle velocity, $V_3$, can be large. This measurement error in
$V_3$ is denoted as $\delta V_3$. The distribution of this error is
determined by the astronomical instrumentation relevant to the
observations of the galaxy at hand and are usually known to the
astronomer. In the implementation of the methodology to real and
simulated data, as discussed below, we work with a Gaussian error
distribution with a known variance $\sigma^2_{V_3}$. Thus, $\delta
V_3\sim {\cal N}(0, \sigma^2_{V_3})$. For this particular error
distribution, the likelihood is defined as {{
\begin{equation}
\Pr(\{(s_1^{(k)}, s_2^{(k)}, v_3^{(k)})^T\}_{k=1}^{N_{data}}\vert\brho,{\bf F}) =
{\displaystyle{\prod_{k=1}^{N_{data}} \nu(s_1^{(k)}, s_2^{(k)}, v_3^{(k)})\ast\left[\frac{1}{\sigma_{v_3^{(k)}}}\exp\left(\frac{-(v_3^{(k)})^2}{2\sigma^2_{v_3^{(k)}}}\right)\right]}}.
\end{equation}}}
For any other distribution of the uncertainties in the measurement of
$V_3$, the likelihood is to be rephrased as resulting from a
convolution of $\nu(\cdot,\cdot,\cdot)$ and that chosen error
distribution.

\subsection{Priors}
\label{sec:priors}
\noindent
In the existing astronomical literature, there is nothing to suggest
the $pdf$ of the state space variable in a real galaxy though there are theoretical
models of the functional dependence between stellar energy ($E$) and
angular momentum ($L$) and $pdf$ of $\bS$ and $\bV$ \ctp{BT}. Given this, we
opt for uniform priors on $f_{c,d}$, $c=1,2,\ldots,N_\epsilon$,
$d=1,2,\ldots,N_\ell$. However, in our inference, we will use the
suggestion of monotonicity of the state space $pdf$, as given in the
theoretical galactic dynamics literature. We also use the physically
motivated constraint that $f_{c,d}\geq 0$, $\forall\:c,d$. Thus, we
use $f_{c,d}\sim{\cal U}(1,0)$, where ${\cal U}(\cdot,\cdot)$ denotes
the uniform distribution over the interval $[\cdot,\cdot]$.

As far as priors on the gravitational mass density are concerned,
astronomical models are available \ctp{BT}. All such models suggest that
gravitational mass density is a monotonically decreasing function of $R$.
A numerically motivated
form that has been used in the astrophysical community is referred to
as the NFW density \ctp{NFW}, though criticism of predictions obtained
with this form also exist \ctp[among others]{deblok2003}. For our purpose we suggest a
uniform prior on $\rho_b$ such that 
\begin{eqnarray}
\pi_0(\rho_b) &=& \displaystyle{\frac{1}{\Upsilon_{hi}^{(b)}(R_s,\rho_0) -\Upsilon_{lo}^{(b)}(R_s,\rho_0)}}\quad{\mbox{where}} \nonumber \\
\Upsilon_{lo}^{(b)}(R_s,\rho_0) &=& 10^{-3} \rho_{NFW}^{(b)}(R_s,\rho_0) \nonumber \\ 
\Upsilon_{hi}^{(b)}(R_s,\rho_0) &=& 10^{3} \rho_{NFW}^{(b)}(R_s,\rho_0)\quad{\mbox{with}} \nonumber\\
\rho_{NFW}^{(b)}(R_s,\rho_0) &:=& \displaystyle{\frac{\rho_0}{\frac{r_b}{R_s}\left(1 + \frac{r_b}{R_s}\right)^2}}, \nonumber \\
r_b &:=& r_{min}+(b-0.5)\delta_r 
\label{eqn:denprior}
\end{eqnarray}
i.e. $\rho_{NFW}^{(b)}(R_s,\rho_0)$ is the gravitational mass density
as given by the 2-parameter NFW form, for the particle radial location
$r\in[r_{min}+(b-1)\delta_r, r_{max}+b\delta_r$,
  $b=1,2,\ldots,N_x$. In fact, this location is summarised as $r_b$,
  the mid-point of the $b$-th radial bin. $R_s$ and $\rho_0$ are the 2
  parameters of the NFW density form. In our work these are
  hyperparameters and we place uniform priors on them:
  $\pi_0(R_s)=1/(r_{max}-r_{min})$ and
  $\pi_0(\rho_0)=1/(10^{14}-10^{9})$, where these numbers are
  experimentally chosen.

\subsection{Posterior}
\label{sec:posterior}
\noindent
Given the data, we use Bayes rule to write down the joint posterior probability density of \\
$\rho_1,\rho_2,\ldots,\rho_{N_x},f_{1,1},\ldots,f_{N_\epsilon,1},f_{1,2},\ldots,f_{N_\epsilon,N_\ell},R_s,\rho_0,N_\ell$. This is 
\begin{eqnarray}
\pi(\brho,{\bf F},R_s,\rho_0,N_\ell\vert {\bu}_1,{\bu}_2,\ldots,\bu_{N_{data}})\propto&& \nonumber \\
\displaystyle{\prod_{k=1}^{N_{data}} \left[\nu(s_1^{(k)}, s_2^{(k)}, v_3^{(k)})\ast\frac{1}{\sigma_{v_3^{(k)}}}\exp\left(\frac{-(v_3^{(k)})^2}{2\sigma^2_{v_3^{(k)}}}\right)\right]}\times &&\nonumber \\      
\displaystyle{\prod_{b=1}^{N_{x}}\left[\frac{1}{\Upsilon_{hi}^{(b)}(R_s,\rho_0) -\Upsilon_{lo}^{(b)}(R_s,\rho_0)}\right]} \times&&  \nonumber \\      
\displaystyle{\frac{1}{r_{max}-r_{min}}}\times\displaystyle{\frac{1}{10^{14}-10^{9}}} \times \displaystyle{\frac{1}{5}}.&&
\end{eqnarray}
where we used $\pi_0(f_{c,d})=1$,
$\forall\:c=1,2,\ldots,N_{\epsilon},\:d=1,2,\ldots,N_\ell$. Here, the
factor
$\displaystyle{\frac{1}{r_{max}-r_{min}}}\times\displaystyle{\frac{1}{10^{14}-10^{9}}}
\times \displaystyle{\frac{1}{5}}$ is a constant and therefore can be
subsumed into the constant of proportionality that defines the above
relation. 

We marginalise $\rho_0$ and $R_s$ out of \\$\pi(\brho,{\bf
  F},R_s,\rho_0,N_\ell\vert {\bu}_1,{\bu}_2,\ldots,\bu_{N_{data}})$ to
achieve the joint posterior probability of $\brho$, ${\bf F}$ and
$N_\ell$ given the data. The marginalisation involves only the term
$\displaystyle{\prod_{b=1}^{N_{x}}\left[\frac{1}{\Upsilon_{hi}^{(b)}(R_s,\rho_0)
      -\Upsilon_{lo}^{(b)}(R_s,\rho_0)}\right]}=$\\$\displaystyle{\prod_{b=1}^{N_{x}}\left[\frac{(r_{min}-0.5\delta_r+b\delta_r)(R_s+r_{min}-0.5\delta_r+b\delta_r)^2}{\rho_0
      R_s^3(10^3-10^{-3})}\right]}$ (recalling Equation~\ref{eqn:denprior}). Integrating this term over a
fixed interval of values of $R_s$ and again over a fixed interval of
$\rho_0$, result in a constant that depends on $N_{data}$, $r_{min}$
and $\delta_r$. Thus the marginalisation only results in a constant
that can be subsumed within the unknown constant of proportionality
that we do not require the exact computation of, given that posterior
samples are generated using adaptive Metropolis-Hastings \ctp{haario}. Thus we can write down the joint posterior probability of $\brho$, ${\bf F}$ and $N_\ell$ given the data as:
{{
\begin{equation}
\pi(\brho,{\bf F},N_\ell\vert {\bu}_1,{\bu}_2,\ldots,\bu_{N_{data}})\propto
\displaystyle{\prod_{k=1}^{N_{data}} \left[\nu(s_1^{(k)}, s_2^{(k)}, v_3^{(k)})\ast\frac{1}{\sigma_{v_3^{(k)}}}\exp\left(\frac{-(v_3^{(k)})^2}{2\sigma^2_{v_3^{(k)}}}\right)\right]} 
\label{eqn:posteriorfinal}
\end{equation}
}}
We discuss the implemented inference next.

\section{Inference}
\label{sec:inference}
\noindent
We intend to make inference on each component of the vector $\brho$
and the matrix ${\bf F}$, along with $N_\ell$. We do this under the
constraints of a gravitational mass density function $\rho(R)$ that is
non-increasing with $R$ and a $pdf$ $f(E,L)$ of the state space variable that is
non-increasing with $E$. Motivation for these constraints is presented
in Section~\ref{sec:priors}. In other words, $\rho_b \geq \rho_{b+1}$ and
$f_{c,d} \leq f_{c+1,d}$ for $b=1,2,\ldots,N_x$ and $\rho_{N_x+1}:=0$.
Also, here $c=1,2,\ldots,N_\epsilon-1$ and $d=1,2,\ldots,N_\ell$. 

First we discuss performing inference on $\brho$ using adaptive
Metropolis-Hastings \ctp{haario}, while maintaining this constraint of
monotonicity. We define
\begin{eqnarray}
\Delta_b:= \rho_{b}-\rho_{b+1}, \quad b=1,2,\ldots,N_{x},&&\nonumber \\
\mbox{with}\quad\rho_{N_x+1}:=0.&&
\end{eqnarray}
It is on the parameters $\Delta_1, \Delta_2,\ldots, \Delta_{N_x-1}$
that we make inference. Let within our inference scheme, at the $n$-th
iteration, the current value of $\Delta_b$ be $\delta_b^{(n)}$. Let in
this iteration, a candidate value $\tilde{\delta_b}^{(n)}$ of
$\Delta_b$ be proposed from the folded normal density ${\cal
  N}_{folded}(\mu_b, \sigma^2_b)$, i.e.
\begin{equation}
\tilde{\delta_b}^{(n)} \sim {\cal N}_{folded}(\mu_b, \sigma^2_b)
\end{equation}
where the choice of a folded normal \ctp{folded} or truncated normal proposal
density is preferred over a density that achieves zero probability
mass at the variable value of 0. This is because there is a non-zero
probability for the gravitational mass density to be zero in a given
radial bin. Here $\mu_b$ and $\sigma_b^2$ are the mean and variance of
the proposal density that $\Delta_b$ is proposed from. We choose the
current value of $\Delta_b$ as $\mu_b$ and in this adaptive inference
scheme, the variance is given by the empirical variance of the chain
since the $n_0$-th iteration, i.e.
\begin{eqnarray}
\mu_b &=& \delta_b^{(n)}\nonumber \\
\sigma_b^2 &=& \displaystyle{\frac{\sum_{q=n_0}^{n-1} \left[\delta_b^{(q)}\right]^2}{n-n_0}} - \displaystyle{\left[\frac{\sum_{q=n_0}^{n-1} \delta_b^{(q)}}{n-n_0}\right]^2}
\end{eqnarray}
We choose the folded normal proposal density given its ease of computation:
\begin{equation}
q(\tilde{\delta_b}^{(n)}; \mu_b,\sigma_b^2) = 
\displaystyle{\frac{1}{2\pi\sigma_b}\exp\left[-\frac{(\tilde{\delta_b}^{(n)}-\mu_b)^2}{2\sigma_b^2} -\frac{(\tilde{\delta_b}^{(n)}+\mu_b)^2}{2\sigma_b^2}\right]} 
\end{equation}
It is evident that this is a symmetric proposal density. We discuss
the acceptance criterion in this standard Metropolis-Hastings scheme,
after discussing the proposal density of the components of the matrix
${\bf F}$ and the parameter $N_\ell$.

If $\tilde{\delta_b}^{(n)}$ is accepted, then the updated $b$-th
component of $\brho$ in the $n$-th iteration is ${\rho_b}^{(n)} =
{\rho_{b+1}}^{(n)} + \tilde{\delta_b}^{(n)}$. If the proposed
candidate is rejected then ${\rho_b}^{(n)}$ resorts back to
${\rho_b}^{(n)} = {\rho_{b+1}}^{(n)} + {\delta_b}^{(n)}$. 

Along similar lines, we make inference directly on 
\begin{eqnarray}
\Gamma_{c,d}=f_{c,d}-f_{c+1,d},&&\nonumber \\ 
c=1,2,\ldots,N_\epsilon-1,\:d=1,2,\ldots,N_\ell,&&\nonumber \\
f_{N_\epsilon,d}=0.&&
\end{eqnarray}
Let in the $n$-th iteration, the current value of $\Gamma_{c,d}$ be 
$\gamma_{c,d}^{(n)}$ and the proposed value be 
$\tilde{\gamma}_{c,d}^{(n)}$ where the proposed candidate is sampled
from the folded normal density ${\cal N}_{folded}(\gamma_{c,d}^{(n)},(\tau^{(n)}_{c,d})^2)$ where the variance $(\tau^{(n)}_{c,d})^2$ is again the empirical variance of the chain between the $n_0^{/}$-th and the $n-1$-th iteration. Then
the updated general element of the state space $pdf$ matrix in this iteration is $f^{(n)}_{c,d} = f^{(n)}_{c+1,d}+\tilde{\gamma}_{c,d}^{(n)}$, if the proposed value as accepted, otherwise, $f^{(n)}_{c,d} = f^{(n)}_{c+1,d}+{\gamma_{c,d}}^{(n)}$. Thus, the proposal density that a component of the ${\bf F}$ matrix is proposed from is also symmetric.

We propose $N_\ell$ from the discrete uniform distribution, i.e. the proposed 
value of $N_\ell$ in the $n$-th iteration is 
\begin{equation}
N_\ell \sim {\cal U}_{discrete}[z_1, z_2]
\end{equation}
where the bounds of the interval $[z_1,z_2]$ are found experimentally given the data at hand. 

Given that we are making inference on the $\{\Delta_b\}_{b=1}^{N_x}$
and $\{\Gamma_{c,d}\}_{c=1,d=1}^{N_\epsilon,N_\ell}$, we rephrase the
posterior probability of the unknowns as\\
$\pi(\Delta_1,\ldots,,\Delta_{N_x},\Gamma_{1,1},\ldots,\Gamma_{N_\epsilon,N_\ell},N_\ell\vert\bu_1,\ldots,\bu_{N_{data}})$. This
posterior density is proportional to the RHS of
Equation~\ref{eqn:posteriorfinal}.

Then given that the proposal densities that components of $\brho$ and of ${\bf F}$ are sampled from and that the proposal density for $N_\ell$ is uniform, the Metropolis-Hastings acceptance ratio is reduced to the ratio of the posterior of the proposed state space vector value to that of the current state space vector, i.e. the proposed state space vector
$(\tilde{\Delta_1},\ldots,,\tilde{\Delta}_{N_x},\tilde{\Gamma}_{1,1},\ldots,\tilde{\Gamma}_{N_\epsilon,\tilde{N}_\ell},\tilde{N}_\ell)^T$
is accepted if 
\begin{equation}
\displaystyle{\frac{\pi(\tilde{\Delta_1},\ldots,\tilde{\Delta}_{N_x},\tilde{\Gamma}_{1,1},\ldots,\tilde{\Gamma}_{N_\epsilon,\tilde{N}_\ell},\tilde{N}_\ell\vert \bu_1,\ldots,\bu_{N_{data}})}{\pi({\Delta_1},\ldots,{\Delta_{N_x}},{\Gamma_{1,1}},\ldots,{\Gamma_{N_\epsilon,N_\ell}},N_\ell\vert \bu_1,\ldots,\bu_{N_{data}})}} < u
\end{equation}
where the uniform random variable $u\sim{\cal U}[0,1]$.

\section{Illustration on synthetic data}
\label{sec:synthetic}
\noindent
In this section we illustrate the methodology on synthetic data set simulated from a chosen models for the $pdf$ of $\bW=(\bS^T,\bV^T)^T$. $N_{data}$=198. The chosen models for this $pdf$ are $f_{WD}(E(\bS,\bV),L(\bS,\bV))$ or $f_{WD}(E,L)$ and $f_{Michie}(E,L)$. These are given by:
\begin{eqnarray}                                                              
f_{WD}(E, L) =&& \displaystyle{\frac{1}{\sqrt{2\pi\sigma^2}}}\displaystyle{\exp\left(-\frac{L^2}{r_a\sigma^2}\right)\exp\left(\frac{-E}{\sigma^2}\right)}, \nonumber \\       
f_{Michie}(E, L) =&& 
\displaystyle{\frac{1}{\sqrt{2\pi\sigma^2}}}\displaystyle{\exp\left(-\frac{L^2}{r_a\sigma^2}\right)\left[\exp{\left(\frac{-E}{\sigma^2}\right)} - 1 \right]}, 
\end{eqnarray}                    
where $E(\bS,\bV)=\Phi(\rho_{Model}(R))+V^2/2$ with $\rho_{Model}(R)$
chosen in both models for the state space $pdf$ to be
$\rho_{Model}(R)=\displaystyle{\left(\frac{3M}{4\pi a^3}\right)\left(1
  + \frac{r^2}{a^2}\right)^{-5/2}}$. Here the model parameters $r_a>0$
and $\sigma\>0$ are assigned realistic numerical values. From these 2
chosen $pdf$s, $N_{data}$ values of $\bU$ were sampled; these 2
samples constituted the 2 synthetic data sets ${\bf D}_{WD}$ and ${\bf
  D}_{Michie}$. The learnt gravitational mass density parameters and
discretised version of the state space $pdf$ are displayed in
Figure~\ref{fig:syn1}. Some of the convergence characteristics of the
chains are explored in Figure~\ref{fig:syn2}. The trace of the joint
posterior probability of the unknown parameters given the data is
shown along with histograms of $\rho_2$ learnt from 3 distinct parts
of the chain that is run using data ${\bf D}_{WD}$.

\begin{figure}[!t]                                             
\centering{
  $\begin{array}{c c}                                                         
   \includegraphics[height=.3\textheight]{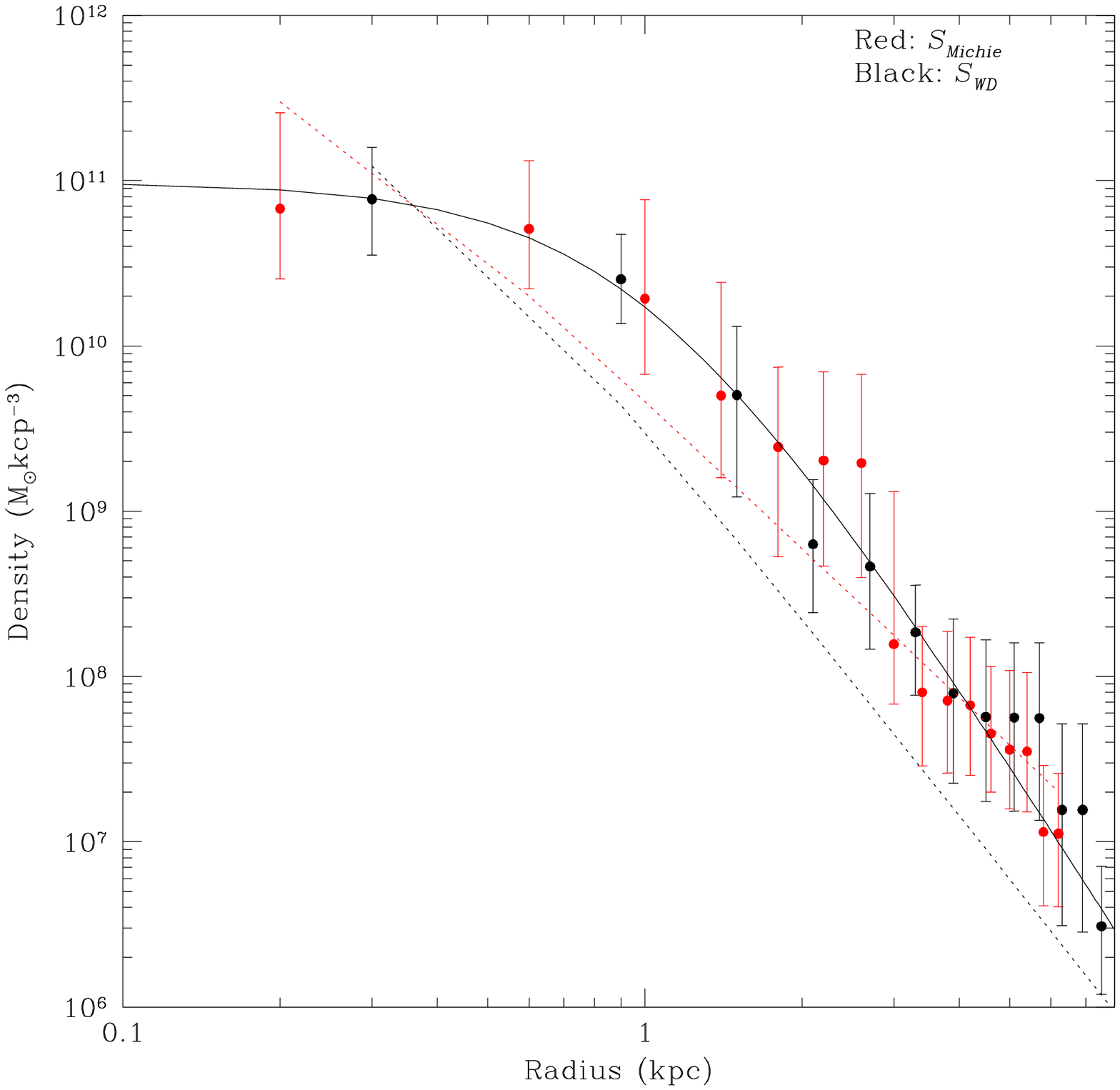} & 
   \includegraphics[height=.3\textheight]{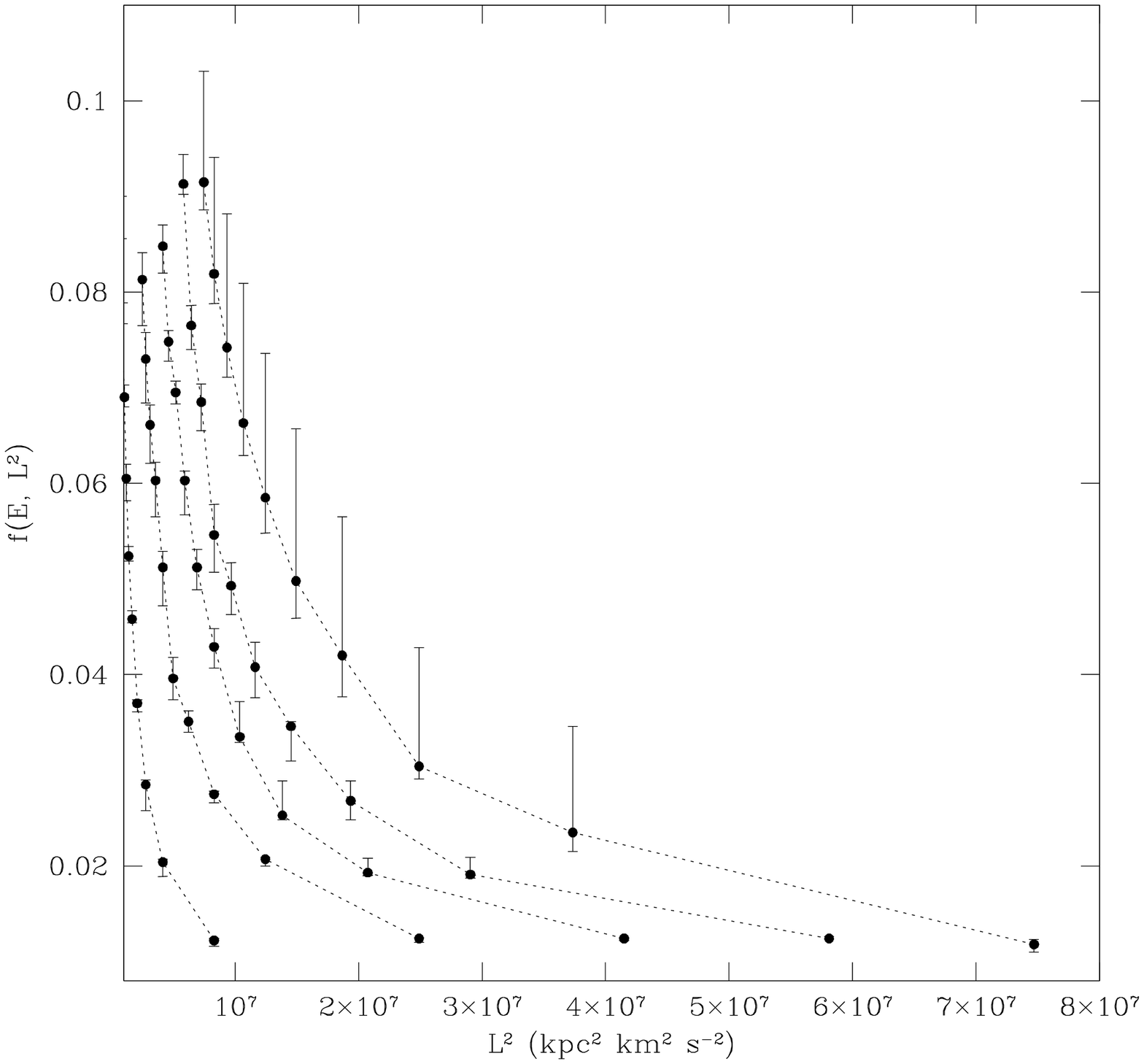}
   \end{array}$
}
\caption{{\it Left:} gravitational mass density parameters learnt
  using synthetic data sets ${\bf D}_{WD}$ and ${\bf D}_{Michie}$ that
  are sampled from the chosen models of the $pdf$ of the state space
  variable, at the chosen model of the gravitational mass density
  function $\rho_{Model}(R)$ which is shown in the black solid
  line. The 95$\%$ highest probability density (HPD) credible region
  is represented as the error bar on each estimated parameter while
  the parameter value at the mode of its marginal posterior probability
  is shown by the filled circle. The density parameters
  $\rho^{(b)}_{NFW}(R_s,\rho_0)$, $b=1,2,\ldots,N_x$, are joined with
  the dotted lines in red and black where the prior on the sought
  parameter $\rho_b$ is defined in terms of
  $\rho^{(b)}_{NFW}(R_s,\rho_0)$ (see
  Equation~\ref{eqn:denprior}). {\it Right}: discretised $pdf$ of
  $\bS$ and $\bV$ learnt using data ${\bf D}_{WD}$, plotted against
  $\ell^2$ i.e. square of the value of $L(\bS,\bV)$, at 5 different
  values of $E(\bS,\bV)$. The true values of the parameters are joined
  in dotted lines. }
\label{fig:syn1}
\end{figure}         

\begin{figure}[!t]
\centering
{
\vspace*{1.3cm}
\includegraphics[height=.32\textheight]{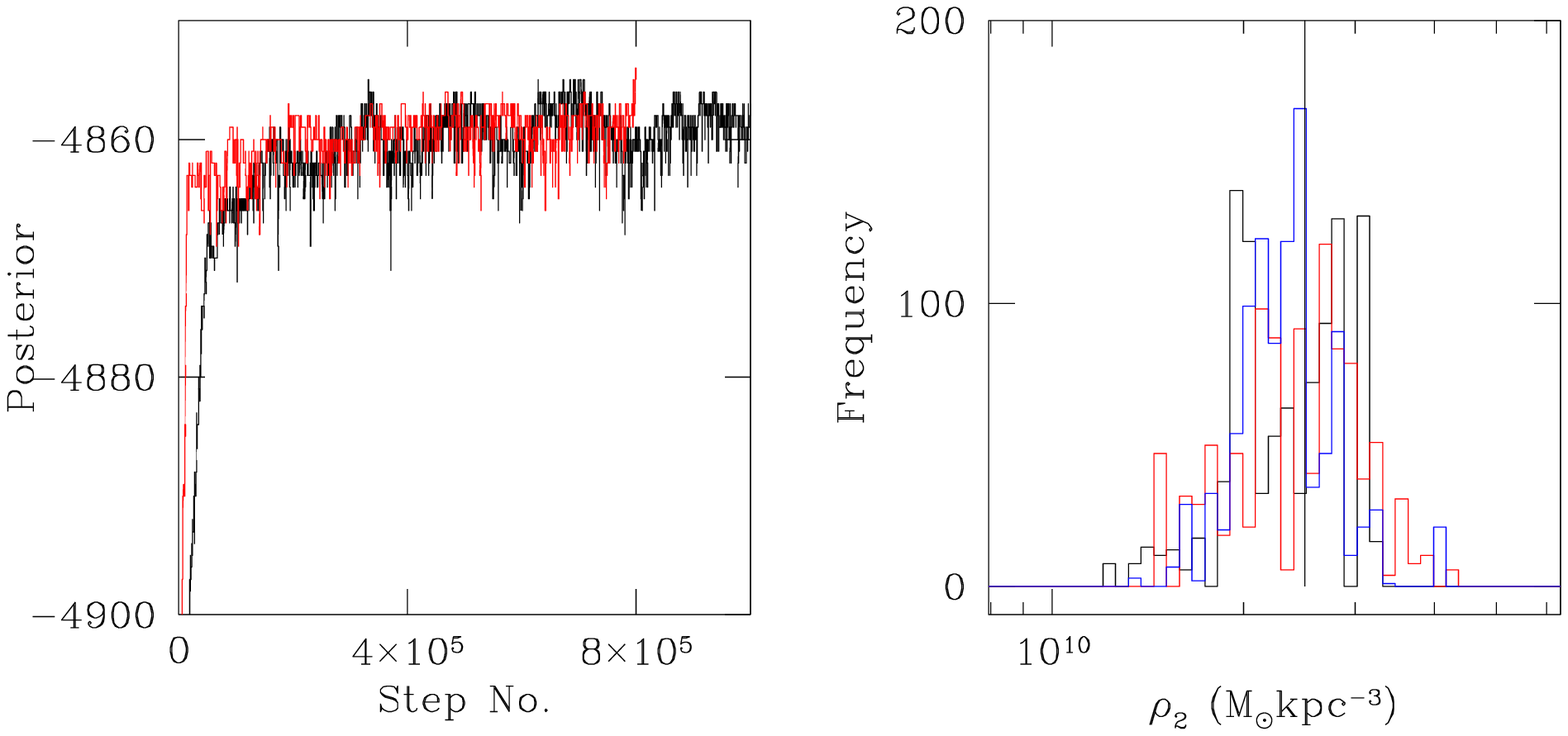}}
\caption{{\it Left}: Trace of the joint posterior probability density of all 
the unknowns, given the synthetic data sets ${\bf D}_{WD}$ and ${\bf D}_{Michie}$, in black and red. {\it Right}: Histograms of values of the parameter $\rho_2$  in 3 equally sized and non-overlapping parts of the chain run with ${\bf D}_{WD}$, where all 3 parts were sampled post burnin, between iteration number 600,000 and 800,000. The true value of $\rho_2$ is marked by the black solid line.\vspace*{1cm}}
\label{fig:syn2}
\end{figure}

\section{Illustration on real data}
\label{sec:real}
\noindent
In this section we present the gravitational mass density parameters
and the state space $pdf$ parameters learnt for the real galaxy
NGC3379 using 2 data sets ${\bf D}_{PNe}$ and ${\bf D}_{GC}$ which respectively
have sample size 164 \ctp{pns} and 29 \ctp{bergond}. An independent test of hypothesis exercise shows that there is relatively higher support in ${\bf D}_{GC}$ for an isotropic $pdf$ of the state space variable $\bW=(\bS^T,\bV^T)^T$ than in ${\bf D}_{PNe}$. Given this, some runs were performed using an isotropic model of the state space $pdf$; this was achieved by fixing the number $N_\ell$ of $L$-bins to 1. Then $L$ identically takes the value $\ell_1$ and is rendered a constant. This effectively implies that the domain of $f(E,L)$ is rendered uni-dimensional, i.e. the state space $pdf$ is then rendered $f(E)$. Recalling the definition of an isotropic function from Remark~\ref{remark:isotropic}, we realise that the modelled state space $pdf$ is then an isotropic function of $\bS$ and $\bV$. Results from chains run with such an isotropic state space $pdf$ were overplotted on results from chains run with the more relaxed version of the $pdf$ that allows for incorporation of anisotropy; in such chain, $N_ell$ is in fact learnt from the data.

\begin{figure}[!h]
\vspace*{1cm}
\begin{center}
\includegraphics[width=11cm]{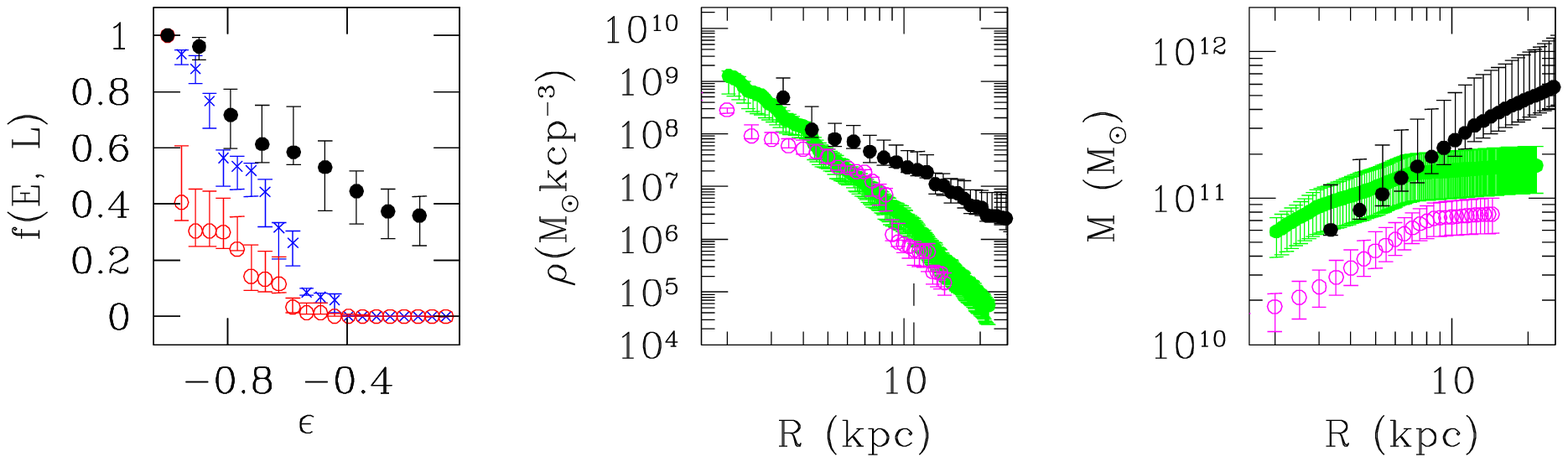}
\end{center}
\caption{{\it Left:} The left panel represents the $f(E,L)$ plotted as
  in red and blue against (the value of $E(\bS,\bV)$) $\epsilon$, at two
  different $\ell$, recovered from a chains that use data ${\bf
    D}_{PNe}$. The modal value of the learnt number of $L$-bins is 7
  for this run. The state space $pdf$ parameters recovered using data
  ${\bf D}_{GC}$ are shown in black. {\it Middle:} Gravitational mass
  density parameters $\rho_i$ estimated from a chain run with ${\bf
    D}_{PNe}$ are shown in magenta, over-plotted on the same obtained
  using the same data, from a chain in which the number of $L$-bins,
  $N_\ell=1$. When $N_\ell$ is fixed as 1, it implies that $L(\bS,\bV)$
  is then no longer a variable and then $f(E,L)$ is effectively
  univariate, depending on $E(\bS,\bV)$ alone. Such a state space
  $pdf$ is an isotropic function of $\bS$ and $\bV$ (see
  Remark~\ref{remark:isotropic}). The $\rho_i$ estimated from such an
  isotropic $pdf$ of the state space variable is shown here in
  green. The mass density parameters learnt using the data ${\bf D}_{GC}$--again learnt from an isotropic state space $pdf$--are
  shown in black. {\it Right:} Figure showing estimates of $M_i =
  \displaystyle{\sum_{j=1}^{i} 4\pi\rho_j\delta_r^2 (j^2 - (j-1)^2)}$,
    against $R$. Here $i=1,2,\ldots,N_x$. The parameters in magenta
    are obtained from the same chain that produce the $\rho_i$
    parameters in the middle panel using ${\bf D}_{Pne}$ while those
    in green and black are obtained using the $\rho_i$ that were
    represented in the middle panel in the corresponding colours.}
\label{fig:anisotropy} 
\end{figure}

\section{Discussions}
\label{sec:discussions}
\noindent
In this work we focused on an inverse problem in which noisy and
partially missing data on the measurable $\bU=(X_1,X_2,V_3)^T$ is used
to make inference on the model parameter vector $\brho$ which is the
discretisation of the unknown model function
$\rho(\parallel\bX\parallel)\equiv\rho(\parallel\bS\parallel)$, where
$\bS$ is an orthogonal transformation of $\bX$ and
$\bX=(X_1,X_2,X_3)^T$. The measurable and the sought function are
related via an unknown function. Given that the very Physics that
connects $\bU$ to $\rho(R)$ is unknown--where
$R:=\parallel\bS\parallel$--we cannot construct training data,
i.e. data comprising a set of computed $\bu$ for a known $\rho(r)$. In
the absence of training data, we are unable to learn the unknown
functional relationship between data and model function, either using
splines/wavelets or by modelling this unknown function with a Gaussian
Process. We then perform the estimation of $\rho(R)$ at chosen values
of $R$, i.e. discretise the range of values of $R$ and estimate
the vector $\brho$ instead, where $\rho_i$ is the value of $\rho(r)$
for $r$ in the $i$-th $R$-bin. We aim to write the posterior of
$\brho$ given the data. The likelihood could be written as the product
of the values of the $pdf$ of the state space vector
$\bW=(\bX^T,\bV^T)^T$ achieved at each data point, but the data being
missing, the $pdf$ is projected onto the space of $\bU$ and the
likelihood is written in terms of these projections of the
$pdf$. $\brho$ is embedded within the definition of the domain of the
$pdf$ of $\bW$. The projection calls for identification of the mapping
between this domain and the unobserved variables $X_3, V_1, V_2$; this
is an application specific task. The likelihood is convolved with the
error distribution and vague but proper priors are invoked, leading to
the posterior probability of the unknowns given the data. Inference is
performed using adaptive MCMC. The method is used to learn the
gravitational mass density of a simulated galaxy using synthetic data,
as well as that in the real galaxy NGC3379, using data of 2 different
kinds of galactic particles. The gravitational mass density vector
estimated from the 2 independent data sets are found to be distinct.

The distribution of the gravitational mass in the system is indicated by the function $M(r)=\displaystyle{\int_{r'=0}^r 4\pi (r')^2 \rho(r') dr'}$. the discretised form of this function defines the parameters $M_i$, $i=1,2,\ldots,N_x$. These are computed using the learnt value of the $\rho_i$ parameters and plotted in Figure~\ref{fig:anisotropy}. We notice that the estimate of $\rho_i$ can depend on the model chosen for the state space $pdf$; thus, the same galaxy can be inferred to be characterised by a higher gravitational mass distribution depending on whether an isotropic state space is invoked or not. Turning this result around, one can argue that in absence of priors on how isotropic the state space of a galaxy really is, the learnt gravitational mass density function might give an erroneous indication of how much gravitational mass there is in this galaxy and of corse how that mass is distributed. It may be remarked that in lieu of such prior knowledge about the topology of the system state space, it is best to consider the least constrained of models for the state space $pdf$, i.e. to consider this $pdf$ to be dependent on both $E(\bS,\bV)$  and $L(\bS,\bV)$.

It is also to be noted that the estimate for the gravitational mass density in the real galaxy NGC3379 appears to depend crucially on which data set is being implemented in the estimation exercise. It is possible that the underlying $pdf$ of the variable $\bW=(\bS^T,\bV^T)^T$ is different for the sub-volume of state space that one set of data vectors are sampled from, compared to another. As these data vectors are components of $\bS$ and $\bV$ of different kinds of galactic particles, this implies that the state space $pdf$ that the different kinds of galactic particles relax into, are different.

\renewcommand\baselinestretch{1.3}
\normalsize


\end{document}